\documentclass[acmsmall]{acmart}
\acmJournal{PACMPL}
\acmVolume{}
\acmNumber{} 
\acmArticle{1} 
\acmYear{2018}
\acmMonth{3} 
\acmDOI{}
\startPage{1}
\acmPrice{} 

\setcopyright{none}

\usepackage{booktabs}   

\usepackage{subcaption} 
\usepackage[utf8]{inputenc} 
\usepackage{latexsym}
\usepackage{amssymb}
\usepackage{stmaryrd}
\usepackage{proof}
\usepackage{xcolor}
\newcommand{\blue}[1]{{\color{blue}#1}}
\newcommand{\red}[1]{{\color{red}#1}}
\newcommand{\green}[1]{{\color{green}#1}}
\newcommand{\tick}{\blue{\m{tick}}}
\newcommand{\delay}{\red{\m{delay}}}
\newcommand{\when}[1]{\red{\m{when?}\;#1}}
\newcommand{\now}[1]{\red{\m{now!}\;#1}}
\newcommand{\noww}{\red{\m{now!}}}
\newcommand{\whenn}{\red{\m{when?}}}

\newcommand{\vdashi}{\vdash^{\!\!\scriptscriptstyle i}}
\newenvironment{sill}{\begin{quote}\begin{tabbing}}{\end{tabbing}\end{quote}}
\newcommand{\tock}{`}
\newcommand{\wsubt}{<:}

\newcommand{\frank}[1]{\green{[[Frank says: #1]]}}

\newcommand{\m}[1]{\mathsf{#1}}

\newcommand{\CC}{\mathcal{C}}
\newcommand{\DD}{\mathcal{D}}
\newcommand{\EE}{\mathcal{E}}

\newcommand{\palign}[1]{\raisebox{-0.75em}{#1}}
\newcommand{\mi}[1]{\mbox{\it #1}}
\newcommand{\mc}[1]{\mathcal{#1}}

\newcommand{\uup}[2]{{{}^{#1}_{#2}}{{\uparrow}\kern-0.2em{\downarrow}}}
\newcommand{\ddown}[2]{{{}^{#1}_{#2}}{{\downarrow}\kern-0.2em{\uparrow}}}

\newcommand{\semi}{\mathrel{;}}

\newcommand{\vvdash}{\mathrel{\vdash\kern-0.8ex\vdash}}

\newcommand{\lolli}{\multimap}
\newcommand{\tensor}{\otimes}
\newcommand{\with}{\mathbin{\binampersand}}

\newcommand{\one}{\mathbf{1}}

\newcommand{\Dia}{\Diamond}
\newcommand{\Next}{\raisebox{0.3ex}{$\scriptstyle\bigcirc$}}

\begin{document}

\title[Parallel Complexity Analysis]{Parallel Complexity Analysis with Temporal Session Types}

\author{Ankush Das}
\affiliation{
  \institution{Carnegie Mellon University}
  \country{USA}
}
\email{ankushd@cs.cmu.edu}

\author{Jan Hoffmann}
\affiliation{
  \institution{Carnegie Mellon University}
  \country{USA}
}
\email{jhoffmann@cmu.edu}

\author{Frank Pfenning}
\affiliation{
  \institution{Carnegie Mellon University}
  \country{USA}
}
\email{fp@cs.cmu.edu}

\begin{abstract}
  We study the problem of parametric parallel complexity
analysis of concurrent, message-passing programs.
To make the analysis local and compositional, it is based on a
conservative extension of binary session types, which structure the
type and direction of communication between processes and stand in a
Curry-Howard correspondence with intuitionistic linear logic.
The main innovation is to enrich session types with the \emph{temporal
  modalities} \emph{next} ($\Next A$), \emph{always} ($\Box A$), and
\emph{eventually} ($\Dia A$), to additionally prescribe the timing of
the exchanged messages in a way that is precise yet flexible.
The resulting \emph{temporal session types} uniformly express
properties such as the message rate of a stream, the latency of a
pipeline, the response time of a concurrent queue, or the span of a
fork/join parallel program.
The analysis is parametric in the cost model and the presentation
focuses on communication cost as a concrete example.
The soundness of the analysis is established by proofs of progress and
type preservation using a timed multiset rewriting semantics.
Representative examples illustrate the scope and usability of the
approach.




\end{abstract}

\maketitle

\section{Introduction}
\label{sec:intro}

For sequential programs, several type systems and program analyses
have been proposed to structure,
formalize~\citep{LagoG11,DannerLR15,CicekBGGH16}, and
automate~\citep{GulwaniMC09,HoffmannW15,AvanziniLM15} complexity analysis.
Analyzing the complexity of concurrent, message-passing processes
poses additional challenges that these systems do not address.
To begin with, we need information about the possible interactions
between processes to enable compositional and local reasoning
about concurrent cost.

Session types~\citep{Honda98esop} provide a structured way to
prescribe communication behavior between message-passing processes and
are a natural foundation for compositional, concurrent complexity
analysis. In particular, we use a system of binary session types that
stands in a Curry-Howard correspondence with intuitionistic linear
logic~\citep{Caires10concur,Caires16mscs}. Our communication model is
\emph{asynchronous} in the sense of the asynchronous $\pi$-calculus:
sending always succeeds immediately, while receiving blocks until a
message arrives.

In addition to the structure of communication, the timing of messages
is of central interest for analyzing concurrent cost. With information
on message timing we may analyze not only properties such as the rate
or latency with which a stream of messages can proceed through a
pipeline, but also the span of a parallel computation, which can be
defined as the time of the final response message assuming maximal
parallelism.

There are several possible ways to enrich session types with timing
information. A challenge is to find a balance between precision and
flexibility. We would like to express precise times according to a
global clock as in synchronous dataflow languages whenever that is
possible. However, sometimes this will be too restrictive.  For
example, we may want to characterize the response time of a concurrent
queue where enqueue and dequeue operations arrive at unpredictable
intervals.


In this paper, we develop a type system that captures the parallel
complexity of session-typed message-passing programs by adding
\emph{temporal modalities} \emph{next} ($\Next A$), \emph{always}
($\Box A$), and \emph{eventually} ($\Dia A$), interpreted over a
linear model of time.  
When considered as types, the temporal modalities allow us to express
properties of concurrent programs such as the \emph{message rate} of a
stream, the \emph{latency} of a pipeline, the \emph{response time} of
concurrent data structure, or the \emph{span} of a fork/join parallel
program, all in the same uniform manner. Our results complement prior
work on expressing the \emph{work} of session-typed processes in the
same base language~\citep{Das18arxiv}. Together, they form a
foundation for analyzing 
the parallel implementation complexity of session-typed processes.

The way in which we construct the type system is conservative over the
base language of session types, which makes it quite general and
easily able to accommodate various concrete cost models. Our language
contains standard session types and process expressions, and their
typing rules remain unchanged. They correspond to processes that do
not induce cost and send all messages at the same constant time $0$.

To model computation cost we introduce a new syntactic form $\delay$,
which advances time by one step.
To specify a particular cost semantics we take an ordinary,
non-temporal program and add delays capturing the intended cost.  For
example, if we decide only the blocking operations should cost one
unit of time, we add a delay before the continuation of every
receiving construct.  If we want sends to have unit cost as well, we
also add a delay immediately after each send operation.
Processes that contain delays cannot be typed using standard session
types.

To type processes with non-zero cost, we first introduce the type
$\Next A$, which is inhabited only by the process expression $(\delay
\semi P)$. This forces time to advance on all channels that $P$ can
communicate along. The resulting types prescribe the \emph{exact}
time a message is sent or received and sender and receiver are
precisely synchronized.

As an example, consider a stream of bits terminated by $\m{\$}$,
expressed as the recursive session type
\begin{sill}
  $\m{bits} = {\oplus}\{\m{b0} : \m{bits}, \m{b1} : \m{bits}, \m{\$} : \one\}$
\end{sill}
where ${\oplus}$ stands for \emph{internal choice} and $\one$ for
\emph{termination}, ending the session.  A simple cost model for
asynchronous communication prescribes a cost of one unit of time for
every receive operation. A stream of bits then needs to delay every
continuation to give the recipient time to receive the message,
expressing a \emph{rate} of one. This can be captured precisely with
the temporal modality $\Next A$:
\begin{sill}
  $\m{bits} = {\oplus}\{\m{b0} : \Next \m{bits}, \m{b1} : \Next \m{bits}, \m{\$} : \Next \one\}$
\end{sill}
A transducer $\mi{neg}$ that negates each bit it receives along
channel $x$ and passes it on along channel $y$ would be typed as
\begin{sill}
  $x : \m{bits} \vdash \mi{neg} :: (y : \Next \m{bits})$
\end{sill}
expressing a \emph{latency} of one.  A process $\mi{negneg}$ that puts two
negations in sequence has a latency of two, compared with $\mi{copy}$
which passes on each bit, and $\mi{id}$ which terminates and
identifies the channel $y$ with the channel $x$, short-circuiting the
communication.
\begin{sill}
  $x : \m{bits} \vdash \mi{negneg} :: (y : \Next \Next \m{bits})$ \\
  $x : \m{bits} \vdash \mi{copy} :: (y : \Next \m{bits})$ \\
  $x : \m{bits} \vdash \mi{id} :: (y : \m{bits})$
\end{sill}
All these processes have the same extensional behavior, but different
latencies. They also have the same rate since after the
pipelining delay, the bits are sent at the same rate they are
received, as expressed in the common type $\m{bits}$ used
in the context and the result.

While precise and minimalistic, the resulting system is often too
precise for typical concurrent programs such as pipelines or servers.
We therefore introduce the dual type formers $\Dia A$ and $\Box A$ to
talk about varying time points in the future.  Remarkably, even if
part of a program is typed using these constructs, we can still make
precise and useful statements about other aspects.

For example, consider a transducer $\mi{compress}$ that shortens a
stream by combining consecutive 1 bits so that, for example,
$00110111$ becomes $00101$.  For such a transducer, we cannot bound
the latency statically, even if the bits are received at a constant
rate like in the type $\m{bits}$. So we have to express that after
seeing a 1 bit we will \emph{eventually} see either another bit or the
end of the stream.  For this purpose, we introduce a new type
$\m{sbits}$ with the same message alternatives as $\m{bits}$, but
different timing.  In particular, after sending $\m{b1}$ we have
to send either the next bit or end-of-stream \emph{eventually}
($\Dia \m{sbits}$), rather than immediately.
\begin{sill}
  $\m{sbits} = {\oplus}\{\m{b0} : \Next \m{sbits}, \m{b1} : \Next \Dia \m{sbits}, \m{\$} : \Next \one\}$ \\
  $x : \m{bits} \vdash \mi{compress} :: (y : \Next \m{sbits})$
\end{sill}
We write $\Next \Dia \m{sbits}$ instead of $\Dia \m{sbits}$ for the
continuation type after $\m{b1}$ to express that there will always
be a delay of at least one; to account for the unit cost of receive
in this particular cost model.

The dual modality, $\Box A$, is useful to express, for example, that a
server providing $A$ is \emph{always} ready, starting from ``now''.
As an example, consider the following temporal type of an interface to
a process of type $\Box \m{queue}_A$ with elements of type $\Box A$.
It expresses that there must be at least three time units between
successive enqueue operations and that the response to a dequeue
request is immediate, only one time unit later ($\with$ stands for
external choice, the dual to internal choice).
\begin{sill}
  $\m{queue}_{A} = {\with}\{$ \= $\m{enq} : \Next (\Box A \lolli \Next^3 \Box \m{queue}_A),$ \\
  \> $\m{deq} : \Next {\oplus}\{$ \= $\m{none} : \Next \one,$ \\
  \>\> $\m{some} : \Next (\Box A \tensor \Next \Box \m{queue}_A)\,\}\,\}$
\end{sill}
As an example of a \emph{parametric} cost analysis, we can give the
following type to a process that appends inputs $l_1$ and $l_2$ to
yield $l$, where the message rate on all three lists is $r+2$ units of
time (that is, the interval between consecutive list elements needs to
be at least 2).
\begin{sill}
  $l_1 : \m{list}_A[n], l_2 : \Next^{(r+4)n+2}\, \m{list}_A[k]
  \vdash \mi{append} :: (l : \Next \Next \m{list}_A[n+k])$
\end{sill}
It expresses that $\mi{append}$ has a latency of two units of time and
that it inputs the first message from $l_2$ after $(r+4)n+2$ units of
time, where $n$ is the number of elements sent along $l_1$.

To analyze the span of a fork/join parallel program, we capture the
time at which the (final) answer is sent. For example, the type
$\m{tree}[h]$ describes the span of a process that computes
the parity of a binary tree of height $h$ with boolean values at the
leaves. The session type expresses that the result of the computation
is a single boolean that arrives at time $5h+3$ after the $\m{parity}$
request.
\begin{sill}
$\m{tree}[h] = {\with}\{\,\m{parity} : \Next^{5h+3}\, \m{bool}\,\}$
\end{sill}

In summary, the main contributions of the paper are (1) a generic
framework for parallel cost analysis of asynchronously communicating
session-typed processes rooted in a novel combination of temporal and
linear logic, (2) a soundness proof of the type system with
respect to a timed operational semantics, showing progress and
type preservation (3) instantiations of the framework with different
cost models, e.g. where either just
receives, or receives and sends, cost one time unit each, and (4)
examples illustrating the scope of our method. Our technique for
proving progress and preservation does not require dependency graphs
and may be of independent interest.  We further provide decidable
systems for \emph{time reconstruction} and \emph{subtyping} that
greatly simplify the programmer's task. They also enhance modularity
by allowing the same program to be assigned temporally different
types, depending on the context of use.

Related is work on space and time complexity analysis of interaction
nets by \citet{GimenezPOPL16}, which is a parallel execution model for
functional programs.  While also inspired by linear logic and, in
particular, proof nets, it treats only special cases of the additive
connectives and recursive types and does not have analogues of the
$\Box$ and $\Dia$ modalities. It also does not
provide a general source-level programming notation with a
syntax-directed type system. On the other hand they incorporate
sharing and space bounds, which are beyond the scope of this paper.

Another related thread is the research on timed multiparty session
types~\citep{TMSTConcur14} for modular verification of real-time
choreographic interactions.  Their system is based on explicit global
timing interval constraints, capturing a new class of communicating
timed automata, in contrast to our system based on binary session
types in a general concurrent language.  Therefore, their system has
no need for general $\Box$ and $\Dia$ modalities, the ability to pass
channels along channels, or the ability to identify channels via
forwarding. Their work is complemented by an expressive dynamic
verification framework in real-time distributed
systems~\citep{Neykova14beat}, which we do not consider.
Semantics counting communication costs for work and span in
session-typed programs were given by~\citet{Silva16linearity},
but no techniques for analyzing them were provided.

The remainder of the paper is organized as follows.  We review our
basic system of session types in Section~\ref{sec:basic}, then
introduce the next-time modality $\Next A$ in Section~\ref{sec:next}
followed by $\Dia A$ and $\Box A$ in Section~\ref{sec:temporal}.  We
establish fundamental metatheoretic type safety properties in
Section~\ref{sec:metatheory} and time reconstruction in
Section~\ref{sec:recon}.  Additional examples in
Section~\ref{sec:examples} are followed by a discussion of further
related work in Section~\ref{sec:related} and a brief conclusion.

\section{The Base System of Session Types}
\label{sec:basic}


The underlying base system of session types is derived from a
Curry-Howard interpretation of intutionistic linear
logic~\citep{Caires10concur,Caires16mscs}. We present it here to fix
our particular formulation, which can be considered the purely linear
fragment of SILL~\citep{Toninho13esop,Pfenning15fossacs}.
Remarkably,
the rules remain exactly the same when we consider temporal extensions
in the next section.
The key idea is that an intuitionistic linear sequent
\[
A_1, A_2, \ldots, A_n \vdash C
\]
is interpreted as the interface to a \emph{process expression} P.
We label each of the antecedents with a channel name $x_i$
and the succedent with channel name $z$.
The $x_i$'s are \emph{channels used by $P$} and $z$ is
the \emph{channel provided by $P$}.
\[
x_1 : A_1, x_2 : A_2, \ldots, x_n : A_n \vdash P :: (z : C)
\]
The resulting judgment formally states that process $P$ provides
a service of session type $C$ along channel $z$, while using the
services of session types $A_1, \ldots, A_n$ provided along channels
$x_1, \ldots, x_n$ respectively. All these channels must be
distinct, and we sometimes implicitly rename them to preserve this
presupposition.  We abbreviate the antecedent of the
sequent by $\Omega$.


\begin{figure}
\centering
\renewcommand{\arraystretch}{1.3}
\begin{tabular}{l|l|l}
\textbf{Type} & \textbf{Provider Action} & \textbf{Session Continuation} \\\hline
${\oplus}\{\ell : A_\ell\}_{\ell \in L}$ & send label $k \in L$ & $A_k$ \\
${\with}\{\ell : A_\ell\}_{\ell \in L}$ & receive and branch on label $k \in L$ & $A_k$ \\
$\one$ & send token $\m{close}$ & \emph{none} \\
$A \tensor B$ & send channel $c : A$ & $B$ \\
$A \lolli B$ & receive channel $c : A$ & $B$
\end{tabular}
\caption{Basic Session Types.  Every provider action has a matching
  client action.}
\vspace{-1em}
\label{fig:basic-types}
\end{figure}

\begin{figure}
\centering
\renewcommand{\arraystretch}{1.3}
\[
\begin{array}{lcllll}
& & \mbox{\bf Expression} & \mbox{\bf Action} & \mbox{\bf Continuation} & \mbox{\bf Rules} \\
P,Q & ::= & x \leftarrow f \leftarrow \overline{e} \semi Q & \mbox{spawn process named $f$} & [a/x]Q & \m{def} \\
& \mid & x{:}A \leftarrow P \semi Q & \mbox{spawn $[a/x]P$} & [a/x]Q & \m{cut} \\
& \mid & c \leftarrow d & \mbox{identify $c$ and $d$} & \mbox{\it none} & \m{id} \\
& \mid & c.k \semi P & \mbox{send label $k$ along $c$} & P & {\oplus}R, {\with}L \\
& \mid & \m{case}\; c\; (\ell \Rightarrow P_\ell)_{\ell \in L} & \mbox{receive label $k$ along $c$} & P_k &
{\oplus}L, {\with}R \\
& \mid & \m{close}\; c & \mbox{close $c$} & \mbox{\it none} & {\one}R \\
& \mid & \m{wait}\; c \semi P & \mbox{wait for $c$ to close} & P & {\one}L \\
& \mid & \m{send}\; c\; d \semi P & \mbox{send $d$ along $c$} & P & {\tensor}R, {\lolli}L \\
& \mid & x \leftarrow \m{recv}\; c \semi P & \mbox{receive $d$ along $c$} & [d/x]P & {\tensor}L, {\lolli}R
\end{array}
\]
\vspace{-1em}
\caption{Basic Process Expressions}
\vspace{-1em}
\label{fig:basic-exps}
\end{figure}

Figure~\ref{fig:basic-types} summarizes the basic session types and
their actions. The process expression for these actions are shown in
Figure~\ref{fig:basic-exps}; the process typing rules in
Figure~\ref{fig:basic-typing}.  The first few examples (well into
Section~\ref{sec:temporal}) only use internal choice, termination, and
recursive types, together with process definitions and forwarding, so
we explain these in some detail together with their formal operational
semantics.  A summary of all the operational semantics rules can be
found in Figure~\ref{fig:basic-opsem}.

\subsection{Internal Choice}

A type $A$ is said to describe a \emph{session}, which is a particular
sequence of interactions.  As a first type construct we consider
\emph{internal choice} ${\oplus}\{\ell : A_\ell\}_{\ell \in L}$, an
$n$-ary labeled generalization of the linear logic connective
$A \oplus B$.  A process that provides
$x : {\oplus}\{\ell : A_\ell\}_{\ell \in L}$ can send any label
$k \in L$ along $x$ and then continue by providing $x : A_k$.  We
write the corresponding process as $(x.k \semi P)$, where $P$ is the
continuation.  This typing is formalized by the \emph{right rule}
${\oplus}R$ in our sequent calculus. The corresponding client branches
on the label received along $x$ as specified by the \emph{left rule}
${\oplus}L$.
\[
\begin{array}{c}
\infer[{\oplus}R]
  {\Omega \vdash (x.k \semi P) :: (x : {\oplus}\{\ell : A_\ell\}_{\ell \in L})}
  {(k \in L) & \Omega \vdash P :: (x : A_k)}
\hspace{2.5em}
\infer[{\oplus}L]
  {\Omega, x{:}{\oplus}\{\ell : A_\ell\}_{\ell \in L} \vdash \m{case}\;x\;
    (\ell \Rightarrow Q_\ell)_{\ell \in L} :: (z : C)}
  {(\forall \ell \in L) &
    \Omega, x{:}A_\ell \vdash Q_\ell :: (z : C)}
\end{array}
\]
We formalize the operational semantics as a system of \emph{multiset
  rewriting rules}~\cite{Cervesato09ic}.  We introduce semantic
objects $\m{proc}(c, t, P)$ and $\m{msg}(c, t, M)$ which mean that
process $P$ or message $M$ provide along channel $c$ and are at time
$t$.  A \emph{process configuration} is a multiset of such objects,
where any two offered channels are distinct.  Communication is
asynchronous, so that a process $(c.k \semi P)$ sends a message $k$
along $c$ and continues as $P$ without waiting for it to be received.
As a technical device to ensure that consecutive messages on a
channel arrive in order, the sender also creates a fresh continuation
channel $c'$ so that the message $k$ is actually represented as
$(c.k \semi c \leftarrow c')$ (read: send $k$ along $c$ and continue as
$c'$).
\[
\begin{array}{ll}
({\oplus}S) & \m{proc}(c, t, c.k \semi P) \; \mapsto \; \m{proc}(c', t, [c'/c]P), \m{msg}(c, t, c.k \semi c \leftarrow c') \quad \mbox{($c'$ fresh)}
\end{array}
\]
When the message $k$ is received along $c$, we select branch $k$ and
also substitute the continuation channel $c'$ for $c$.
\[
\begin{array}{ll}
({\oplus}C) & \m{msg}(c, t, c.k \semi c \leftarrow c'), \m{proc}(d, t, \m{case}\;c\;(\ell \Rightarrow Q_\ell)_{\ell \in L})
\; \mapsto \; \m{proc}(d, t, [c'/c]Q_k)
\end{array}
\]
The \emph{message} $(c.k \semi c \leftarrow c')$ is just a particular
form of process, where $c \leftarrow c'$ is \emph{identity} or
\emph{forwarding}, explained in Section~\ref{sec:id}. Therefore
no separate typing rules for messages are needed; they can be
typed as processes~\citep{Balzer17icfp}.

In the receiving rule we require the time $t$ of the message and
receiver process to match.  Until we introduce temporal types, this is
trivially satisfied since all actions are considered instantaneous
and processes will always remain at time $t = 0$.

The dual of internal choice is \emph{external choice}
${\with}\{\ell : A_\ell\}_{\ell \in L}$, which just reverses the role
of provider and client and reuses the same process notation.  It is
the $n$-ary labeled generalization of the linear logic connective
$A \with B$.

\begin{figure}
\centering
\[
\begin{array}{c}
\infer[\m{cut}]
  {\Omega, \Omega' \vdash (x{:}A \leftarrow P \semi Q) :: (z : C)}
  {\Omega' \vdash P :: (x:A) & \Omega, x : A \vdash Q :: (z : C)}
\hspace{3em}
\infer[\m{id}]
  {y : A \vdash (x \leftarrow y) :: (x : A)}
  {\mathstrut}
\\[1em]
\infer[{\oplus}R]
  {\Omega \vdash (x.k \semi P) :: (x : {\oplus}\{\ell : A_\ell\}_{\ell \in L})}
  {(k \in L) & \Omega \vdash P :: (x : A_k)}
\hspace{3em}
\infer[{\oplus}L]
  {\Omega, x{:}{\oplus}\{\ell : A_\ell\}_{\ell \in L} \vdash \m{case}\;x\;
    (\ell \Rightarrow Q_\ell)_{\ell \in L} :: (z : C)}
  {(\forall \ell \in L) &
    \Omega, x{:}A_\ell \vdash Q_\ell :: (z : C)}
\\[1em]
\infer[{\with}R]
  {\Omega \vdash \m{case}\; x\; (\ell \Rightarrow P_\ell)_{\ell \in L} :: (x : {\with}\{\ell : A_\ell\}_{\ell \in L})}
  {(\forall \ell \in L)
   & \Omega \vdash P_\ell :: (x : A_\ell)}
\hspace{2em}
\infer[{\with}L]
  {\Omega, x{:}{\with}\{\ell : A_\ell\}_{\ell \in L} \vdash (x.k \semi Q) :: (z : C)}
  {\Omega, x{:}A_k \vdash Q :: (z : C)}
\\[1em]
\infer[{\one}R]
  {\cdot \vdash (\m{close}\; x) :: (x : \one)}
  {\mathstrut}
\hspace{3em}
\infer[{\one}L]
  {\Omega, x{:}\one \vdash (\m{wait}\; x \semi Q) :: (z : C)}
  {\Omega \vdash Q :: (z : C)}
\\[1em]
\infer[{\tensor}R]
  {\Omega, y{:}A \vdash (\mbox{send}\; x\; y \semi P) :: (x : A \tensor B)}
  {\Omega \vdash P :: (x : B)}
\hspace{3em}
\infer[{\tensor}L]
  {\Omega, x{:}A \tensor B \vdash (y \leftarrow \m{recv}\; x \semi Q) :: (z : C)}
  {\Omega, y{:}A, x{:}B \vdash Q :: (z : C)}
\\[1em]
\infer[{\lolli}R]
  {\Omega \vdash (y \leftarrow \m{recv}\; x \semi P) :: (x : A \lolli B)}
  {\Omega, y{:}A \vdash P :: (x : B)}
\hspace{3em}
\infer[{\lolli}L]
  {\Omega, x{:}A \lolli B, y{:}A \vdash (\m{send}\; x\; y \semi Q) :: (z : C)}
  {\Omega, x{:}B \vdash Q :: (z : C)}
\\[1em]
\infer[\m{def}]
  {\Omega, \Omega' \vdash (x \leftarrow f \leftarrow \Omega' \semi Q) :: (z : C)}
  {(\Omega' \vdash f = P_f :: (x : A)) \in \Sigma
    & \Omega, x{:}A \vdash Q :: (z : C)}
\end{array}
\]
\vspace{-1em}
\caption{Basic Typing Rules}
\vspace{-1em}
\label{fig:basic-typing}
\end{figure}

\subsection{Termination}

The type $\one$, the multiplicative unit of linear logic, represents
termination of a process, which (due to linearity) is not allowed to
use any channels.
\[
\infer[{\one}R]
  {\cdot \vdash \m{close}\; x :: (x : \one)}
  {\mathstrut}
\hspace{3em}
\infer[{\one}L]
  {\Omega, x{:}\one \vdash (\m{wait}\; x \semi Q) :: (z : C)}
  {\Omega \vdash Q :: (z : C)}
\]
Operationally, a client has to wait for the corresponding closing
message, which has no continuation since the provider terminates.
\[
\begin{array}{llcl}
({\one}S) & \m{proc}(c, t, \m{close}\; c) & \mapsto & \m{msg}(c, t, \m{close}\; c) \\
({\one}C) & \m{msg}(c, t, \m{close}\; c), \m{proc}(d, t, \m{wait}\; c \semi Q)
& \mapsto & \m{proc}(d, t, Q)
\end{array}
\]

\subsection{Forwarding}
\label{sec:id}

A process $x \leftarrow y$ \emph{identifies} the channels
$x$ and $y$ so that any further communication along either $x$ or $y$
will be along the unified channel. Its typing rule corresponds to the
logical rule of \emph{identity}.
\[
\infer[\m{id}]
  {y : A \vdash (x \leftarrow y) :: (x : A)}
  {}
\]
We have already seen this form in the continuations of message
objects.  Operationally, the intuition is realized by
\emph{forwarding}: a process $c \leftarrow d$ \emph{forwards} any
message $M$ that arrives along $d$ to $c$ and vice versa.
Because channels are used linearly the forwarding process
can then terminate, making sure to apply the proper renaming.
The corresponding rules of operational semantics are as follows. 
\[
\begin{array}{llcll}
(\m{id}^+C) & \m{msg}(d, t, M), \m{proc}(c, s, c \leftarrow d) & \mapsto & \m{msg}(c, t, [c/d]M) & (t \geq s) \\
(\m{id}^-C) & \m{proc}(c, s, c \leftarrow d), \m{msg}(e, t, M(c)) & \mapsto & \m{msg}(e, t, [d/c]M(c))
& (s \leq t)
\end{array}
\]
In the last transition, we write $M(c)$ to indicate that $c$ must
occur in $M$, which implies that this message is the sole client of
$c$.  In anticipation of the extension by temporal operators, we do
not require the time of the message and the forwarding process to be
identical, but just that the forwarding process is ready \emph{before}
the message arrives.

\subsection{Process Definitions}

Process definitions have the form $\Omega \vdash f = P :: (x : A)$
where $f$ is the name of the process and $P$ its definition. All
definitions are collected in a fixed global signature $\Sigma$.  We
require that $\Omega \vdash P :: (x : A)$ for every definition,
which allows the definitions to be mutually recursive.  For
readability of the examples, we break a definition into two
declarations, one providing the type and the other the process definition
binding the variables $x$ and those in $\Omega$ (generally omitting
their types):
\[
\begin{array}{l}
\Omega \vdash f :: (x : A) \\
x \leftarrow f \leftarrow \Omega = P
\end{array}
\]
A new instance of a defined process $f$ can be spawned with
the expression
\[
x \leftarrow f \leftarrow \overline{y} \semi Q
\]
where $\overline{y}$ is a sequence of variables matching the
antecedents $\Omega$.  The newly spawned process will use all
variables in $\overline{y}$ and provide $x$ to the continuation $Q$.
The operational semantics is defined by
\[
\begin{array}{ll}
(\m{def}C) &
\m{proc}(c, t, x \leftarrow f \leftarrow \overline{e} \semi Q)
\;\mapsto\;
\m{proc}(a, t, [a/x, \overline{e}/\Omega]P),
\m{proc}(c, t, [a/x]Q)
\quad (\mbox{$a$ fresh})
\end{array}
\]
Here we write $\overline{e}/\Omega$ to denote substitution of the
channels in $\overline{e}$ for the corresponding variables in
$\Omega$.

Sometimes a process invocation is a \emph{tail call}, written without
a continuation as $x \leftarrow f \leftarrow \overline{y}$.  This is a
short-hand for
$x' \leftarrow f \leftarrow \overline{y} \semi x \leftarrow x'$ for a
fresh variable $x'$, that is, we create a fresh channel and
immediately identify it with $x$ (although it is generally implemented
more efficiently).

\begin{figure}
\centering
\[
\begin{array}{lll}
(\m{cut}C) & \m{proc}(c, t, x{:}A \leftarrow P \semi Q) \;\mapsto\; \m{proc}(a, t, [a/x]P), \m{proc}(c, t, [a/x]Q) & \mbox{($a$ fresh)} \\[1ex]
(\m{def}C) &
\m{proc}(c, t, x \leftarrow f \leftarrow \overline{e} \semi Q)
\;\mapsto\;
\m{proc}(a, t, [a/x, \overline{e}/\Omega_f]P_f),
\m{proc}(c, t, [a/x]Q)
& (\mbox{$a$ fresh}) \\[1ex]
(\m{id}^+C) & \m{msg}(d, t, M), \m{proc}(c, s, c \leftarrow d) \;\mapsto\; \m{msg}(c, t, [c/d]M) & (t \geq s) \\
(\m{id}^-C) & \m{proc}(c, s, c \leftarrow d), \m{msg}(e, t, M(c)) \;\mapsto\; \m{msg}(e, t, [d/c]M(c)) & (s \leq t) \\[1ex]
({\oplus}S) & \m{proc}(c, t, c.k \semi P) \;\mapsto\; \m{proc}(c', t, [c'/c]P), \m{msg}(c, t, c.k \semi c \leftarrow c') & \mbox{($c'$ fresh)} \\
({\oplus}C) & \m{msg}(c, t, c.k \semi c \leftarrow c'), \m{proc}(d, t, \m{case}\;c\;(\ell \Rightarrow Q_\ell)_{\ell \in L})
\;\mapsto\; \m{proc}(d, t, [c'/c]Q_k) \\[1ex]
({\with}S) & \m{proc}(d, t, c.k \semi Q) \;\mapsto\; \m{msg}(c', t, c.k \semi c' \leftarrow c), \m{proc}(d, t, [c'/c]Q) & \mbox{($c'$ fresh)} \\
({\with}C) & \m{proc}(c, t, \m{case}\;c\;(\ell \Rightarrow Q_\ell)_{\ell \in L}), \m{msg}(c', t, c.k \semi c' \leftarrow c)
\;\mapsto\; \m{proc}(c', t, [c'/c]Q_k) \\[1ex]
({\one}S) & \m{proc}(c, t, \m{close}\; c) \;\mapsto\; \m{msg}(c, t, \m{close}\; c) \\
({\one}C) & \m{msg}(c, t, \m{close}\; c), \m{proc}(d, t, \m{wait}\; c \semi Q) \;\mapsto\; \m{proc}(d, t, Q)\\[1ex]
({\tensor}S) & \m{proc}(c, t, \m{send}\; c\; d \semi P) \;\mapsto\; \m{proc}(c', t, [c'/c]P), \m{msg}(c, t, \m{send}\; c\; d \semi c \leftarrow c')
& \mbox{($c'$ fresh)} \\
({\tensor}C) & \m{msg}(c, t, \m{send}\; c\; d \semi c \leftarrow c'), \m{proc}(e, t, x \leftarrow \m{recv}\; c \semi Q)
\;\mapsto\; \m{proc}(e, t, [c', d/c, x]Q) \\[1ex]
({\lolli}S) & \m{proc}(e, t, \m{send}\; c\; d \semi Q) \;\mapsto\; \m{msg}(c', t, \m{send}\; c\; d \semi c' \leftarrow c), \m{proc}(e, t, [c'/c]Q) & \mbox{($c'$ fresh)} \\
({\lolli}C) & \m{proc}(c, t, x \leftarrow \m{recv}\; x \semi P), \m{msg}(c', t, \m{send}\; c\; d \semi c' \leftarrow c)
\;\mapsto\; \m{proc}(c', t, [c', d/c, x]P)
\end{array}
\]
\vspace{-1em}
\caption{Basic Operational Semantics}
\vspace{-1em}
\label{fig:basic-opsem}
\end{figure}

\subsection{Recursive Types}
\label{sec:rec-types}

Session types can be naturally extended to include recursive types.
For this purpose we allow (possibly mutually recursive) type
definitions $X = A$ in the signature, where we require $A$ to be
\emph{contractive}~\citep{Gay05acta}.  This means here that $A$ should
not itself be a type name.  Our type definitions are
\emph{equi-recursive} so we can silently replace $X$ by $A$ during
type checking, and no explicit rules for recursive types are needed.

As a first example, consider a stream of bits (introduced in Section
\ref{sec:intro}) defined recursively as
\begin{sill}
  $\m{bits} = {\oplus}\{\m{b0} : \m{bits}, \m{b1} : \m{bits}, \m{\$} : \one\}$
\end{sill}
When considering bits as representing natural numbers, we think of the
least significant bit being sent first.  For example, a process
$\mi{six}$ sending the number $6 = (110)_2$ would be
\begin{sill}
$\cdot \vdash \mi{six} :: (x : \m{bits})$ \\
$x \leftarrow \mi{six} = x.\m{b0} \semi x.\m{b1} \semi x.\m{b1} \semi x.\m{\$} \semi \m{close}\; x$
\end{sill}
Executing $\m{proc}(c_0, 0, c_0 \leftarrow \mi{six})$ yields
(with some fresh channels $c_1, \ldots, c_4$)
\[
\begin{array}{lcl}
\m{proc}(c_0, 0, c_0 \leftarrow \m{six})
& \mapsto^* &
\m{msg}(c_4, 0, \m{close}\; c_4), \\
 & &\m{msg}(c_3, 0, c_3.\m{\$} \semi c_3 \leftarrow c_4), \\
 & &\m{msg}(c_2, 0, c_2.\m{b1} \semi c_2 \leftarrow c_3), \\
 & &\m{msg}(c_1, 0, c_1.\m{b1} \semi c_1 \leftarrow c_2), \\
 & &\m{msg}(c_0, 0, c_0.\m{b0} \semi c_0 \leftarrow c_1)
\end{array}
\]
As a first example of a recursive process definition, consider one
that just copies the incoming bits.
\begin{sill}
$y : \m{bits} \vdash \mi{copy} :: (x : \m{bits})$ \\
$x \leftarrow \mi{copy} \leftarrow y = $ \\
\quad $\m{case}\; y\;$ \= $(\,\m{b0} \Rightarrow $ \= $x.\m{b0} \semi
x \leftarrow \mi{copy} \leftarrow y$ 
\hspace{1em}\= \% received $\m{b0}$ on $y$, send $\m{b0}$ on $x$, recurse \\
\> $\mid \m{b1} \Rightarrow x.\m{b1} \semi
x \leftarrow \mi{copy} \leftarrow y$ 
\>\> \% received $\m{b1}$ on $y$, send $\m{b1}$ on $x$, recurse \\
\> $\mid \m{\$} \Rightarrow x.\m{\$} \semi \m{wait}\; y \semi \m{close}\; x\,)$
\>\> \% received $\m{\$}$ on $y$,
send $\m{\$}$ on $x$, wait on $y$, close $x$
\end{sill}
The process $\m{neg}$ mentioned in the introduction would just swap
the occurrences of $x.\m{b0}$ and $x.\m{b1}$.  We see here an
occurrence of a (recursive) \emph{tail call} to $\mi{copy}$.

A last example in this section: to increment a bit stream we turn
$\m{b0}$ to $\m{b1}$ but then forward the remaining bits unchanged
($x \leftarrow y$), or we turn $\m{b1}$ to $\m{b0}$ but then increment
the remaining stream ($x \leftarrow \mi{plus1} \leftarrow y$) to
capture the effect of the carry bit.
\begin{sill}
$y : \m{bits} \vdash \mi{plus1} :: (x : \m{bits})$ \\
$x \leftarrow \mi{plus1} \leftarrow y = $ \\
\quad $\m{case}\; y\;$ \= $(\,\m{b0} \Rightarrow $ \= $x.\m{b1} \semi
x \leftarrow y$ \\
\> $\mid \m{b1} \Rightarrow x.\m{b0} \semi
x \leftarrow \mi{plus1} \leftarrow y$ \\
\> $\mid \m{\$} \Rightarrow x.\m{\$} \semi \m{wait}\; y \semi \m{close}\; x\,)$
\end{sill}

\section{The Temporal Modality Next ($\Next A$)}
\label{sec:next}


In this section we introduce \emph{actual cost} by explicitly
advancing time. Remarkably, all the rules we have presented so far
remain literally unchanged. As mentioned, they correspond to the
cost-free fragment of the language in which time never advances.  In
addition, we have a new type construct $\Next A$ (read: \emph{next
  $A$}) with a corresponding process construct $(\delay \semi P)$,
which advances time by one unit. In the corresponding typing rule
\[
\infer[{\Next}LR]
  {\Next \Omega \vdash (\delay \semi P) :: (x : \Next A)}
  {\Omega \vdash P :: (x : A)}
\]
we abbreviate $y_1{:}\Next B_1, \ldots, y_m{:}\Next B_m$ by
$\Next (y_1{:}B_1, \ldots, y_m{:}B_m)$.  Intuitively, when
$(\delay \semi P)$ idles, time advances on \emph{all} channels
connected to $P$.  Computationally, we delay the process for one time
unit without any external interactions.
\[
\begin{array}{ll}
({\Next}C) & \m{proc}(c, t, \delay \semi P) \; \mapsto \; \m{proc}(c, t+1, P)
\end{array}
\]
There is a subtle point about forwarding: A process
$\m{proc}(c, t, c \leftarrow d)$ may be ready to forward a message
\emph{before} a client reaches time $t$ while in all other rules the
times must match exactly. We can avoid this mismatch by transforming
uses of forwarding $x \leftarrow y$ at type $\Next^n S$ where
$S \not= \Next(-)$ to $(\delay^n \semi x \leftarrow y)$. In this
discussion we have used the following notation which will be useful
later:
\[
\begin{array}{lcl@{\hspace{3em}}lcl}
\Next^0 A & = & A & \delay^0 \semi P & = & P \\
\Next^{n+1} A & = & \Next \Next^n A & \delay^{n+1} \semi P & = & \delay \semi \delay^n \semi P \\
\end{array}
\]

\subsection{Modeling a Cost Semantics}

Our system allows us to represent a variety of different abstract cost
models in a straightforward way.  We will mostly use two different
abstract cost models.  In the first, called $\mc{R}$, we assign unit
cost to every receive action while all other operations remain
cost-free. We may be interested in this since receiving a message is
the only blocking operation in the asynchronous semantics. A second
one, called $\mc{RS}$ and considered in Section~\ref{sec:examples},
assigns unit cost to both send and receive actions.

To capture $\mc{R}$ we take a source program and insert a delay
operation before the continuation of every receive.  We write this
delay as $\tick$ in order to remind the reader that it arises
systematically from the cost model and is never written by the
programmer.  In all other respects, $\tick$ is just a synonym for
$\delay$.

For example, the earlier copy process would become
\begin{sill}
$\m{bits} = {\oplus}\{\m{b0} : \m{bits}, \m{b1} : \m{bits}, \m{\$} : \one\}$ \\[1ex]
$y : \m{bits} \vdash \mi{copy} :: (x : \m{bits})$ \hspace{4em} 
\% \emph{No longer correct!}\\
$x \leftarrow \mi{copy} \leftarrow y = $ \\
\quad $\m{case}\; y\;$ \= $(\,\m{b0} \Rightarrow \tick \semi x.\m{b0} \semi
x \leftarrow \mi{copy} \leftarrow y$ \\
\> $\mid \m{b1} \Rightarrow \tick \semi x.\m{b1} \semi
x \leftarrow \mi{copy} \leftarrow y$ \\
\> $\mid \m{\$} \Rightarrow \tick \semi x.\m{\$} \semi \m{wait}\; y \semi \tick \semi \m{close}\; x\,)$
\end{sill}
As indicated in the comment, the type of $\mi{copy}$ is now no longer correct because
the bits that arrive along $y$ are delayed by one unit before they are
sent along $x$.  We can observe this concretely by starting to
type-check the first branch
\begin{sill}
$y : \m{bits} \vdash \mi{copy} :: (x : \m{bits})$ \\
$x \leftarrow \mi{copy} \leftarrow y = $ \\
\quad $\m{case}\; y\;$ \= $(\,\m{b0} \Rightarrow \null$ \= \hspace{6em} \=\% $y : \m{bits} \vdash x : \m{bits}$ \\
\>\> $\tick \semi \ldots)$
\end{sill}
We see that the delay $\tick$ does not type-check, because neither $x$
nor $y$ have a type of the form $\Next(-)$.  We need to redefine the
type $\m{bits}$ so that the continuation type after every label
is delayed by one, anticipating the time it takes to receive the label
$\m{b0}$, $\m{b1}$, or $\m{\$}$.  Similarly, we capture in the type of
$\mi{copy}$ that its \emph{latency} is one unit of time.
\begin{sill}
$\m{bits} = {\oplus}\{\m{b0} : \Next \m{bits}, \m{b1} : \Next \m{bits}, \m{\$} : \Next \one\}$ \\[1ex]
$y : \m{bits} \vdash \mi{copy} :: (x : \Next\m{bits})$
\end{sill}
With these declarations, we can now type-check the definition of
$\mi{copy}$.  We show the intermediate type of the used and
provided channels after each interaction.\pagebreak[4]
\begin{sill}
$x \leftarrow \mi{copy} \leftarrow y = $ \\
\quad $\m{case}\; y\;$ \= $(\,\m{b0} \Rightarrow \null$ \= \hspace{9em} \=\% $y : \Next\m{bits} \vdash x : \Next\m{bits}$ \\
\>\> $\tick \semi$ \>\% $y : \m{bits} \vdash x : \m{bits}$ \\
\>\> $x.\m{b0} \semi$ \>\% $y : \m{bits} \vdash x : \Next\m{bits}$ \\ 
\>\> $x \leftarrow \mi{copy} \leftarrow y$ \>\% \emph{well-typed by type of $\mi{copy}$} \\
\> $\mid \m{b1} \Rightarrow \null$ \>\>\% $y : \Next\m{bits} \vdash x : \Next\m{bits}$ \\
\>\> $\tick \semi$ \> \% $y : \m{bits} \vdash x : \m{bits}$ \\
\>\> $x.\m{b1} \semi$ \> \% $y : \m{bits} \vdash x : \Next\m{bits}$ \\
\>\> $x \leftarrow \mi{copy} \leftarrow y$ \\
\> $\mid \m{\$} \Rightarrow$ \>\>\% $y : \Next\one \vdash x : \Next\m{bits}$ \\
\>\> $\tick \semi$ \>\% $y : \one \vdash x : \m{bits}$ \\
\>\> $x.\m{\$} \semi$ \>\% $y : \one \vdash x : \Next\one$ \\
\>\> $\m{wait}\; y \semi$ \>\% $\cdot \vdash x : \Next\one$ \\
\>\> $\tick \semi$ \>\% $\cdot \vdash x : \one$ \\
\>\> $\m{close}\; x\,)$
\end{sill}

Armed with this experience, we now consider the increment process
$\mi{plus1}$.  Again, we expect the latency of the increment to be one
unit of time. Since we are interested in detailed type-checking, we
show the transformed program, with a delay $\tick$ after each receive.
\begin{sill}
$\m{bits} = {\oplus}\{\m{b0} : \Next \m{bits}, \m{b1} : \Next \m{bits}, \m{\$} : \Next \one\}$ \\[1ex]
$y : \m{bits} \vdash \mi{plus1} :: (x : \Next \m{bits})$ \\
$x \leftarrow \mi{plus1} \leftarrow y = $ \\
\quad $\m{case}\; y\;$ \= $(\,\m{b0} \Rightarrow $ \= $\tick \semi x.\m{b1} \semi
x \leftarrow y$ \hspace{7em}\=\% \emph{type error here!} \\
\> $\mid \m{b1} \Rightarrow \tick \semi x.\m{b0} \semi
x \leftarrow \mi{plus1} \leftarrow y$ \\
\> $\mid \m{\$} \Rightarrow \tick \semi x.\m{\$} \semi \m{wait}\; y \semi \tick \semi \m{close}\; x\,)$
\end{sill}
The branches for $\m{b1}$ and $\m{\$}$ type-check as before, but the
branch for $\m{b0}$ does not.  We make the types at the crucial point
explicit:
\begin{sill}
$x \leftarrow \mi{plus1} \leftarrow y = $ \\
\quad $\m{case}\; y\;$ \= $(\,\m{b0} \Rightarrow $ \= $\tick \semi x.\m{b1} \semi$
\hspace{5em}\=\% $y : \m{bits} \vdash x : \Next \m{bits}$ \\
\>\> $x \leftarrow y$ \>\% \emph{ill-typed, since $\m{bits} \not= \Next \m{bits}$} \\
\> $\mid \ldots\,)$
\end{sill}
The problem here is that identifying $x$ and $y$ removes the delay
mandated by the type of $\mi{plus1}$.  A solution is to call $\mi{copy}$
to reintroduce the latency of one time unit.
\begin{sill}
$y : \m{bits} \vdash \mi{plus1} :: (x : \Next \m{bits})$ \\
$x \leftarrow \mi{plus1} \leftarrow y = $ \\
\quad $\m{case}\; y\;$ \= $(\,\m{b0} \Rightarrow $ \= $\tick \semi x.\m{b1} \semi
x \leftarrow \mi{copy} \leftarrow y$ \\
\> $\mid \m{b1} \Rightarrow \tick \semi x.\m{b0} \semi
x \leftarrow \mi{plus1} \leftarrow y$ \\
\> $\mid \m{\$} \Rightarrow \tick \semi x.\m{\$} \semi \m{wait}\; y \semi \tick \semi \m{close}\; x\,)$
\end{sill}

In order to write $\mi{plus2}$ as a pipeline of two increments we need
to delay the second increment explicitly in the program and stipulate,
in the type, that there is a latency of two.
\begin{sill}
$y : \m{bits} \vdash \mi{plus2} :: (x : \Next\Next\m{bits})$ \\
$x \leftarrow \mi{plus2} \leftarrow y =$ \\
\quad\= $z \leftarrow \mi{plus1} \leftarrow y \semi$ \hspace{5em}\=\% $z : \Next\m{bits} \vdash x : \Next\Next\m{bits}$ \\
\> $\delay \semi$ \>\% $z : \m{bits} \vdash x : \Next\m{bits}$ \\
\> $x \leftarrow \mi{plus1} \leftarrow z$
\end{sill}
Programming with so many explicit delays is tedious, but fortunately
we can transform a source program without all these delay operations
(but explicitly temporal session types) automatically in two steps:
(1) we insert the delays mandated by the cost model (here: a $\tick$
after each receive), and (2) we perform \emph{time reconstruction} to
insert the additional delays so the result is temporally well-typed or
issue an error message if this is impossible (see
Section~\ref{sec:recon}).

\subsection{The Interpretation of a Configuration}

We reconsider the program to produce the number $6 = (110)_2$ under
the cost model from the previous section where each receive action
costs one unit of time.
There are no receive operations in this program, but time
reconstruction must insert a delay after each send in order to match
the delays mandated by the type $\m{bits}$.
\begin{sill}
$\m{bits} = {\oplus}\{\m{b0} : \Next\m{bits}, \m{b1} : \Next\m{bits}, \m{\$} : \Next\one\}$ \\[1ex]
$\cdot \vdash \mi{six} :: (x : \m{bits})$ \\
$x \leftarrow \mi{six} = x.\m{b0} \semi \delay \semi x.\m{b1} \semi \delay \semi x.\m{b1} \semi \delay \semi x.\m{\$} \semi \delay \semi \m{close}\; x$
\end{sill}
Executing $\m{proc}(c_0, 0, c_0 \leftarrow \mi{six})$ then leads to the
following configuration
\[
\begin{array}{l}
\m{msg}(c_4, 4, \m{close}\; c_4), \\
\m{msg}(c_3, 3, c_3.\m{\$} \semi c_3 \leftarrow c_4), \\
\m{msg}(c_2, 2, c_2.\m{b1} \semi c_2 \leftarrow c_3), \\
\m{msg}(c_1, 1, c_1.\m{b1} \semi c_1 \leftarrow c_2), \\
\m{msg}(c_0, 0, c_0.\m{b0} \semi c_0 \leftarrow c_1) \\
\end{array}
\]
These messages are at increasing times, which means any client of
$c_0$ will have to immediately (at time 0) receive $\m{b0}$, then (at
time 1) $\m{b1}$, then (at time 2) $\m{b1}$, etc.  In other words, the
time stamps on messages predict \emph{exactly} when the message will
be received.  Of course, if there is a client in parallel we may never
reach this state because, for example, the first $\m{b0}$ message
along channel $c_0$ may be received before the continuation of the
sender produces the message $\m{b1}$.  So different configurations may
be reached depending on the \emph{scheduler} for the concurrent
processes.
It is also possible to give a time-synchronous semantics in which all
processes proceed \emph{in parallel} from time 0 to time 1, then from
time 1 to time 2, etc. 

\section{The Temporal Modalities Always ($\Box A$) and Eventually ($\Dia A$)}
\label{sec:temporal}


The strength and also the weakness of the system so far is that its
timing is very precise. Now consider a process $\mi{compress}$ that
combines runs of consecutive 1's to a single 1.  For example,
compressing $11011100$ should yield $10100$.  First, in the cost-free
setting we might write
\begin{sill}
$\m{bits} = {\oplus}\{\m{b0} : \m{bits}, \m{b1} : \m{bits}, \m{\$} : \one\}$ \\[1ex]
$y : \m{bits} \vdash \mi{compress} :: (x : \m{bits})$ \\
$y : \m{bits} \vdash \mi{skip1s} :: (x : \m{bits})$ \\[1ex]
$x \leftarrow \mi{compress} \leftarrow y =$ \\
\quad $\m{case}\; y\;$ \= $(\,\m{b0} \Rightarrow x.\m{b0} \semi x \leftarrow \mi{compress} \leftarrow y$ \\
\>$\mid \m{b1} \Rightarrow x.\m{b1} \semi x \leftarrow \mi{skip1s} \leftarrow y$ \\
\>$\mid \m{\$} \Rightarrow x.\m{\$} \semi \m{wait}\; y \semi \m{close}\; x\,)$ \\[1ex]
$x \leftarrow \mi{skip1s} \leftarrow y =$ \\
\quad $\m{case}\; y\;$ \= $(\,\m{b0} \Rightarrow x.\m{b0} \semi x \leftarrow \mi{compress} \leftarrow y$ \\
\>$\mid \m{b1} \Rightarrow x \leftarrow \mi{skip1s} \leftarrow y$ \\
\>$\mid \m{\$} \Rightarrow x.\m{\$} \semi \m{wait}\; y \semi \m{close}\; x\,)$
\end{sill}
The problem is that if we adopt the cost model $\mc{R}$ where every
receive takes one unit of time, then this program cannot be typed.
Actually worse: there is no way to insert next-time modalities into
the type and additional delays into the program so that the result is
well-typed.  This is because if the input stream is unknown we cannot
predict how long a run of 1's will be, but the length of such a run
will determine the delay between sending a bit 1 and the following bit
0.

The best we can say is that after a bit 1 we will \emph{eventually}
send either a bit 0 or the end-of-stream token $\m{\$}$.  This is the
purpose of the type $\Dia A$.  We capture this timing in the type
$\m{sbits}$ (for \emph{slow bits}).
\begin{sill}
$\m{bits} = {\oplus}\{\m{b0} : \Next \m{bits}, \m{b1} : \Next \m{bits}, \m{\$} : \Next \one\}$ \\
$\m{sbits} = {\oplus}\{\m{b0} : \Next \m{sbits}, \m{b1} : \Next \Dia \m{sbits}, \m{\$} : \Next \one\}$ \\[1ex]
$y : \m{bits} \vdash \mi{compress} :: (x : \Next \m{sbits})$ \\
$y : \m{bits} \vdash \mi{skip1s} :: (x : \Next \Dia \m{sbits})$
\end{sill}
In the next section we introduce the process constructs and typing
rules so we can revise our $\mi{compress}$ and $\mi{skip1s}$ programs so
they have the right temporal semantics.

\subsection{Eventually $A$}
\label{sec:eventually}

A process providing $\Dia A$ promises only that it will eventually
provide $A$. There is a somewhat subtle point here: since not every
action may require time and because we do not check termination
separately, $x : \Dia A$ expresses only that \emph{if the process
  providing $x$ terminates} it will eventually provide $A$.  Thus, it
expresses nondeterminism regarding the (abstract) \emph{time} at which
$A$ is provided, rather than a strict liveness property.  Therefore,
$\Dia A$ is somewhat weaker than one might be used to from
LTL~\citep{Pnueli77}. When restricted to a purely logical fragment,
without unrestricted recursion, the usual meaning is fully restored so
we feel our terminology is justified.  Imposing termination, for
example along the lines of \citet{Fortier13csl} or
\citet{Toninho14tgc} is an interesting item for future work but not
necessary for our present purposes.


When a process offering $c : \Dia A$ is ready, it will send a $\noww$
message along $c$ and then continue at type $A$. Conversely, the
client of $c : \Dia A$ will have to be ready and waiting for the
$\noww$ message to arrive along $c$ and then continue at type $A$.  We
use $(\when{c} \semi Q)$ for the corresponding client. These explicit
constructs are a conceptual device and may not need to be part of an
implementation. They also make type-checking processes entirely
syntax-directed and trivially decidable.

The typing rules for $\noww$ and $\whenn$ are somewhat subtle.
\[
\infer[{\Dia}R]
  {\Omega \vdash (\now{x} \semi P) :: (x : \Dia A)}
  {\Omega \vdash P :: (x : A)}
\hspace{3em}
\infer[{\Dia}L]
  {\Omega, x{:}\Dia A \vdash (\when{x} \semi Q) :: (z : C)}
  {\Next^* \Box \Omega' = \Omega
    & \Omega, x{:} A \vdash Q :: (z : C)
    & C = \Next^* \Dia C'}
\]
The ${\Dia}R$ rule just states that, without constraints, we can at
any time decide to communicate along $x : \Dia A$ and then continue
the session at type $A$.  The ${\Dia}L$ rule expresses that the
process must be ready to receive a $\noww$ message along $x : \Dia A$,
but there are two further constraints.  Because the process
$(\when{x} \semi Q)$ may need to wait an indefinite period of time,
the rule must make sure that communication along $z$ and any channel
in $\Omega$ can also be postponed an indefinite period of time.  We
write $C = \Next^* \Dia C'$ to require that $C$ may be delayed a fixed
finite number of time steps and then must be allowed to communicate at
an arbitrary time in the future.  Similarly, for every channel $y : B$
in $\Omega$, $B$ must have the form $\Next^* \Box B$, where $\Box$
(as the dual of $\Dia$) is introduced in Section~\ref{sec:box}.

In the operational semantics, the central restriction is that
$\whenn$ is ready \emph{before} the $\noww$ message arrives so
that the continuation can proceed immediately as promised by the type.
\[
\begin{array}{lll}
({\Dia}S) & \m{proc}(c, t, \now{c} \semi P) \;\; \mapsto \;\; \m{proc}(c', t, [c'/c]P), \m{msg}(c, t, \now{c} \semi c \leftarrow c') & \mbox{($c'$ fresh)} \\
({\Dia}C) & \m{msg}(c, t, \now{c} \semi c \leftarrow c'), \m{proc}(d, s, \when{c} \semi Q)
\;\; \mapsto \;\; \m{proc}(d, t, [c'/c]Q) & (t \geq s)
\end{array}
\]
We are now almost ready to rewrite the $\mi{compress}$ process in our
cost model $\mc{R}$.  First, we insert $\tick$ before all the actions
that must be delayed according to our cost model. Then we insert
appropriate additional $\delay$, $\whenn$, and $\noww$ actions.  While
$\mi{compress}$ turns out to be straightforward, $\mi{skip1s}$ creates
a difficulty after it receives a $\m{b1}$:
\begin{sill}
$\m{bits} = {\oplus}\{\m{b0} : \Next \m{bits}, \m{b1} : \Next \m{bits}, \m{\$} : \Next \one\}$ \\
$\m{sbits} = {\oplus}\{\m{b0} : \Next \m{sbits}, \m{b1} : \Next \Dia \m{sbits}, \m{\$} : \Next \one\}$ \\[1ex]
$y : \m{bits} \vdash \mi{compress} :: (x : \Next \m{sbits})$ \\
$y : \m{bits} \vdash \mi{skip1s} :: (x : \Next \Dia \m{sbits})$ \\[1ex]
$x \leftarrow \mi{compress} \leftarrow y =$ \\
\quad $\m{case}\; y\;$ \= $(\,\m{b0} \Rightarrow \tick \semi x.\m{b0} \semi x \leftarrow \mi{compress} \leftarrow y$ \\
\>$\mid \m{b1} \Rightarrow \tick \semi x.\m{b1} \semi x \leftarrow \mi{skip1s} \leftarrow y$ \\
\>$\mid \m{\$} \Rightarrow \tick \semi x.\m{\$} \semi \m{wait}\; y \semi \tick \semi \m{close}\; x\,)$ \\[1ex]
$x \leftarrow \mi{skip1s} \leftarrow y =$ \\
\quad $\m{case}\; y\;$ \= $(\,\m{b0} \Rightarrow \tick \semi \now{x} \semi x.\m{b0} \semi x \leftarrow \mi{compress} \leftarrow y$ \\
\>$\mid \m{b1} \Rightarrow $ \= $\tick \semi $ \hspace{11em}\=\% $y : \m{bits} \vdash x : \Dia \m{sbits}$ \\
\>\> $x' \leftarrow \mi{skip1s} \leftarrow y \semi$ \>\% $x' : \Next \Dia \m{sbits} \vdash x : \Dia \m{sbits}$ \\
\>\> $x \leftarrow \mi{idle} \leftarrow x'$ \>\% with $x':\Next \Dia \m{sbits} \vdash \mi{idle} :: (x : \Dia \m{sbits})$ \\
\>$\mid \m{\$} \Rightarrow \tick \semi \now{x} \semi x.\m{\$} \semi \m{wait}\; y \semi \tick \semi \m{close}\; x\,)$
\end{sill}
At the point where we would like to call $\mi{skip1s}$ recursively, we have
\begin{sill}
$y : \m{bits} \vdash x : \Dia \m{sbits}$ \\
but \quad $y : \m{bits} \vdash \mi{skip1s} :: (x : \Next \Dia \m{sbits})$
\end{sill}
which prevents a tail call since
$\Next \Dia \m{sbits} \not= \Dia \m{sbits}$.  Instead we call
$\mi{skip1s}$ to obtain a new channel $x'$ and then use another
process called $\mi{idle}$ to go from $x' : \Next \Dia \m{sbits}$ to
$x : \Dia \m{sbits}$.  Intuitively, it should be possible to implement
such an idling process: $x : \Dia \m{sbits}$ expresses \emph{at some
  time in the future, including possibly right now} while
$x' : \Next \Dia \m{sbits}$ says \emph{at some time in the future, but
  not right now}.

To type the idling process, we need to generalize the ${\Next}LR$ rule
to account for the interactions of $\Next A$ with $\Box A$ and
$\Dia A$.  After all, they speak about the same underlying model of
time.

\subsection{Interactions of $\Next A$ and $\Dia A$}

Recall the left/right rule for $\Next$:
\[
\infer[{\Next}{LR}]
  {\Next \Omega \vdash (\delay \semi P) :: (x : \Next A)}
  {\Omega \vdash P :: (x : A)}
\]
If the succedent were $x : \Dia A$ instead of $x : \Next A$, we should
still be able to delay since we can freely choose when to interact
along $x$.  We could capture this in the following rule (superseded later
by a more general form of ${\Next}LR$):
\[
\infer[{\Next}{\Dia}]
  {\Next \Omega \vdash (\delay \semi P) :: (x : \Dia A)}
  {\Omega \vdash P :: (x : \Dia A)}
\]
We keep $\Dia A$ as the type of $x$ since we retain the full
flexibility of using $x$ at any time in the future after the initial
delay. We will generalize the rule once more in the
next section to account for interactions with $\Box A$.

With this, we can define and type the idling process
parametrically over $A$:
\begin{sill}
$x':\Next \Dia A \vdash \mi{idle} :: (x : \Dia A)$ \\
$x \leftarrow \mi{idle} \leftarrow x' = \delay \semi x \leftarrow x'$
\end{sill}
This turns out to be an example of subtyping (see
Section~\ref{sec:subtyping}), which means that the programmer actually
will not have to explicitly define or even reference an idling
process.
The programmer simply writes the original $\mi{skip1s}$
process (without referencing the $\mi{idle}$ process) and our
subtyping algorithm will use the appropriate rule
to typecheck it successfully.

\subsection{Always $A$}
\label{sec:box}

We now turn our attention to the last temporal modality, $\Box A$,
which is dual to $\Dia A$.  If a process $P$ provides $x : \Box A$
it means it is ready to receive a $\noww$ message along $x$ at any
point in the future.  In analogy with the typing rules for $\Dia A$,
but flipped to the other side of the sequent, we obtain
\[
\infer[{\Box}R]
  {\Omega \vdash (\when{x} \semi P) :: (x : \Box A)}
  {\Next^* \Box \Omega' = \Omega
    & \Omega \vdash P :: (x : A)}
\hspace{3em}
\infer[{\Box}L]
  {\Omega, x{:}\Box A \vdash (\now{x} \semi Q) :: (z : C)}
  {\Omega, x{:}A \vdash Q :: (z : C)}
\]
The operational rules just reverse the role of provider
and client from the rules for $\Dia A$.
\[
\begin{array}{lll}
({\Box}S) &
\m{proc}(d, t, \now{c} \semi Q) \;\; \mapsto \;\; 
\m{msg}(c', t, \now{c} \semi c' \leftarrow c), \m{proc}(d, t, [c'/c]Q)
& \mbox{($c'$ fresh)} \\
({\Box}C) &
\m{proc}(c, s, \when{c} \semi P), \m{msg}(c', t, \now{c} \semi c' \leftarrow c)
\;\; \mapsto \;\;
\m{proc}(c', t, [c'/c]P) & (s \leq t)
\end{array}
\]

As an example for the use of $\Box A$, and also to introduce a new
kind of example, we specify and implement a counter process that can
receive $\m{inc}$ and $\m{val}$ messages.  When receiving an $\m{inc}$
it will increment its internally maintained counter, when receiving
$\m{val}$ it will produce a finite bit stream representing the current
value of the counter.  In the cost-free setting we have the type
\begin{sill}
$\m{bits} = {\oplus}\{\m{b0} : \m{bits}, \m{b1} : \m{bits}, \m{\$} : \one\}$ \\
$\m{ctr} = {\with}\{\m{inc} : \m{ctr}, \m{val} : \m{bits}\}$
\end{sill}
A counter is implemented by a chain of processes, each holding
one bit (either $\mi{bit0}$ or $\mi{bit1}$) or signaling the end
of the chain ($\mi{empty}$). 
For this purpose we implement three processes:
\begin{sill}
$d : \m{ctr} \vdash \mi{bit0} :: (c : \m{ctr})$ \\
$d : \m{ctr} \vdash \mi{bit1} :: (c : \m{ctr})$ \\
$\cdot \vdash \mi{empty} :: (c : \m{ctr})$ \\[1ex]
$c \leftarrow \mi{bit0} \leftarrow d =$ \\
\quad \= $\m{case}\; c\;$ \= $(\,\m{inc} \Rightarrow c \leftarrow \mi{bit1} \leftarrow d$ \hspace{4em} \= \% increment by continuing as $\mi{bit1}$ \\
\>\> $\mid \m{val} \Rightarrow c.\m{b0} \semi d.\m{val} \semi c \leftarrow d\,)$ \> \% send $\m{b0}$ on $c$, send $\m{val}$ on $d$, identify $c$ and $d$ \\[1ex]
$c \leftarrow \mi{bit1} \leftarrow d =$ \\
\> $\m{case}\; c\;$ \> $(\,\m{inc} \Rightarrow d.\m{inc} \semi c \leftarrow \mi{bit0} \leftarrow d$ \> \% send $\m{inc}$ (carry) on $d$, continue as $\mi{bit1}$ \\
\>\> $\mid \m{val} \Rightarrow c.\m{b1} \semi d.\m{val} \semi c \leftarrow d\,)$ \> \% send $\m{b1}$ on $c$, send $\m{val}$ on $d$, identify $c$ and $d$ \\[1ex]
$c \leftarrow \mi{empty} =$ \\
\> $\m{case}\; c\;$ \> $(\,\m{inc} \Rightarrow$ \= $e \leftarrow \mi{empty} \semi$ \hspace{4.5em} \= \% spawn a new $\mi{empty}$ process with channel $e$ \\
\>\>\> $c \leftarrow \mi{bit1} \leftarrow e$ \> \% continue as $\mi{bit1}$ \\
\>\> $\mid \m{val} \Rightarrow c.\m{\$} \semi \m{close}\; c\,)$ \>\> \% send $\$$ on $c$ and close $c$
\end{sill}
Using our standard cost model $\mc{R}$ we notice a problem: the \emph{carry
  bit} (the $d.\m{inc}$ message sent in the $\mi{bit1}$ process) is
sent only on every other increment received because $\mi{bit0}$
continues as $\mi{bit1}$ \emph{without} a carry, and $\mi{bit1}$
continues as $\mi{bit0}$ \emph{with} a carry.  So it will actually take
$2^k$ increments received at the lowest bit of the counter (which
represents the interface to the client) before an increment reaches
the $k$th process in the chain.  This is not a constant number,
so we cannot characterize the behavior exactly using only
the next time modality.  Instead, we say, from a certain point
on, a counter is always ready to receive either an $\m{inc}$
or $\m{val}$ message.
\begin{sill}
$\m{bits} = {\oplus}\{\m{b0} : \Next \m{bits}, \m{b1} : \Next \m{bits}, \m{\$} : \Next \one\}$ \\
$\m{ctr} = \Box {\with}\{\m{inc} : \Next \m{ctr}, \m{val} : \Next \m{bits}\}$
\end{sill}
In the program, we have ticks mandated by our cost model and some
additional $\delay$, $\whenn$, and $\noww$ actions to satisfy the
stated types.  The two marked lines may look incorrect, but are valid
based on the generalization of the ${\Next}LR$ rule in
Section~\ref{sec:interact}.
\begin{sill}
$d : \Next \m{ctr} \vdash \mi{bit0} :: (c : \m{ctr})$ \\
$d : \m{ctr} \vdash \mi{bit1} :: (c : \m{ctr})$ \\
$\cdot \vdash \mi{empty} :: (c : \m{ctr})$ \\[1ex]
$c \leftarrow \mi{bit0} \leftarrow d =$ \\
\quad \= $\m{case}\; c\;$ \= $(\,\m{inc} \Rightarrow $ \= $\tick \semi$ \hspace{7em}\=\% $d : \m{ctr} \vdash c : \m{ctr}$ \kill
\> $\when{c} \semi$ \>\>\>\% $d : \Next \m{ctr} \vdash c : {\with}\{\ldots\}$ \\
\> $\m{case}\; c\;$ \> $(\,\m{inc} \Rightarrow $ \> $\tick \semi$ \>\% $d : \m{ctr} \vdash c : \m{ctr}$ \\
\>\>\> $c \leftarrow \mi{bit1} \leftarrow d$ \\
\>\> $\mid \m{val} \Rightarrow $ \> $\tick \semi$ \>\% $d : \m{ctr} \vdash c : \m{bits}$ \\
\>\>\> $c.\m{b0} \semi$ \>\% $d : \m{ctr} \vdash c : \Next \m{bits}$ \\
\>\>\> $\now{d} \semi d.\m{val} \semi$ \>\% $d : \Next\m{bits} \vdash c : \Next\m{bits}$ \\
\>\>\> $c \leftarrow d\,)$ \\[1ex]
$c \leftarrow \mi{bit1} \leftarrow d =$ \\
\> $\when{c} \semi$ \>\>\>\% $d : \m{ctr} \vdash c : {\with}\{\ldots\}$ \\
\> $\m{case}\; c\;$ \> $(\,\m{inc} \Rightarrow \tick \semi$ \>\>\% $d : \m{ctr} \vdash c : \m{ctr}$
\hspace{2em}\= \mbox {\it (see Section~\ref{sec:interact})}\\
\>\>\> $\now{d} \semi d.\m{inc} \semi$ \>\% $d : \Next\m{ctr} \vdash c : \m{ctr}$ \\
\>\>\> $c \leftarrow \mi{bit0} \leftarrow d$ \\
\>\> $\mid \m{val} \Rightarrow \tick \semi$ \>\>\% $d : \m{ctr} \vdash c : \m{bit}$ \> \mbox  {\it (see Section~\ref{sec:interact})} \\
\>\>\> $c.\m{b1} \semi$ \>\% $d : \m{ctr} \vdash c : \Next \m{bits}$ \\
\>\>\> $\now{d} \semi d.\m{val} \semi$ \>\% $d : \Next\m{bits} \vdash c : \Next\m{bits}$ \\
\>\>\> $c \leftarrow d\,)$ \\[1ex]
$c \leftarrow \mi{empty} =$ \\
\> $\when{c} \semi$ \>\>\>\% $\cdot \vdash c : {\with}\{\ldots\}$ \\
\> $\m{case}\; c\;$ \> $(\,\m{inc} \Rightarrow$ \> $\tick \semi$ \>\% $\cdot \vdash c : \m{ctr}$ \\
\>\>\> $e \leftarrow \mi{empty} \semi$ \>\% $e : \m{ctr} \vdash c : \m{ctr}$ \\
\>\>\> $c \leftarrow \mi{bit1} \leftarrow e$ \\
\>\> $\mid \m{val} \Rightarrow$ \> $\tick \semi c.\m{\$} \semi$ \>\% $\cdot \vdash c : \Next\one$ \\
\>\>\> $\delay \semi \m{close}\; c\,)$
\end{sill}

\subsection{Interactions Between Temporal Modalities}
\label{sec:interact}

Just as $\Next A$ and $\Dia A$ interacted in the rules since their
semantics is based on the same underlying notion of time, so do
$\Next A$ and $\Box A$.  If we execute a delay, we can allow any
channel of type $\Box A$ that we use and leave its type unchanged
because we are not obligated to communicate along it at any
particular time.  It is a little awkward to formulate this because
among the channels used there may be some of type $\Next B$
and some of type $\Box B$.
\[
\infer[{\Next}]
  {\Box \Omega, \Next \Omega' \vdash (\delay \semi P) :: (x : \Next A)}
  {\Box \Omega, \Omega' \vdash P :: (x : A)}
\]
In the example of $\mi{bit1}$ at the end of the previous section,
we have already seen two lines where this generalization was
crucial, observing that $\m{ctr} = \Box {\with}\{\ldots\}$.

But even this rule does not cover all possibilities, because the
channel $x$ could be of type $\Dia A$.  We introduce a new notation,
writing $[A]_L^{-1}$ and $[A]_R^{-1}$ on types and then extend it to
contexts.  Depending on one's point of view, this can be seen as
stepping forward or backward by one unit of time.
\[
\begin{array}{lcl@{\hspace{2em}}lcl}
[\Next A]_L^{-1} & = & A & [\Next A]_R^{-1} & = & A \\\relax
[\Box A]_L^{-1} & = & \Box A & [\Box A]_R^{-1} & = & \mbox{\it undefined} \\\relax
[\Dia A]_L^{-1} & = & \mbox{\it undefined} & [\Dia A]_R^{-1} & = & \Dia A \\\relax
[S]_L^{-1} & = & \mbox{\it undefined} & [S]_R^{-1} & = & \mbox{\it undefined} \\[1ex]
[x:A]_L^{-1} & = & x:[A]_L^{-1} & [x:A]_R^{-1} & = & x:[A]_R^{-1} \\\relax
[\cdot]_L^{-1} & = & \cdot \\\relax
[\Omega, \Omega']_L^{-1} & = & [\Omega]_L^{-1}, [\Omega']_L^{-1} 
\end{array}
\]
Here, $S$ stands for any basic session type constructor as in
Figure~\ref{fig:basic-types}.  We use this notation in the general
rule ${\Next}{LR}$ which can be found in Figure~\ref{fig:temporal}
together with the final set of rules for $\Box A$ and $\Dia A$.  In
conjunction with the rules in Figure~\ref{fig:basic-typing} this
completes the system of temporal session types where all temporal
actions are explicit.  The rule ${\Next}{LR}$ only applies if both
$[\Omega]_L^{-1}$ and $[x:A]_R^{-1}$ are defined.

We call a type $A$ \emph{patient} if it does not force communication
along a channel $x : A$ at any particular point in time.  Because the
direction of communication is reversed between the two sides of a
sequent, a type $A$ is patient if it has the form $\Next^* \Box A'$ if it is
among the antecedents, and $\Next^* \Dia A'$ if it is in the
succedent.  We write $A\; \m{delayed}^\Box$ and $A\; \m{delayed}^\Dia$
and extend it to contexts $\Omega\; \m{delayed}^\Box$ if for every
declaration $(x : A) \in \Omega$, we have $A\; \m{delayed}^\Box$.

\begin{figure}
  \centering
\[
\begin{array}{c}
\infer[{\Next}{LR}]
  {\Omega \vdash (\delay \semi P) :: (x : A)}
  {[\Omega]_L^{-1} \vdash P :: [x : A]_R^{-1}}
\hspace{3em}
\infer[]
  {\Next^* \Box A\; \m{delayed}^\Box}
  {} 
\hspace{3em}
\infer[]
  {\Next^* \Dia A\; \m{delayed}^\Dia}
  {} 
\\[1em]
\infer[{\Dia}R]
  {\Omega \vdash (\now{x} \semi P) :: (x : \Dia A)}
  {\Omega \vdash P :: (x : A)}
\hspace{3em}
\infer[{\Dia}L]
  {\Omega, x{:}\Dia A \vdash (\when{x} \semi Q) :: (z : C)}
  {\Omega\; \m{delayed}^\Box
    & \Omega, x{:} A \vdash Q :: (z : C)
    & C\; \m{delayed}^\Dia}
\\[1em]
\infer[{\Box}R]
  {\Omega \vdash (\when{x} \semi P) :: (x : \Box A)}
  {\Omega\; \m{delayed}^\Box
    & \Omega \vdash P :: (x : A)}
\hspace{3em}
\infer[{\Box}L]
  {\Omega, x{:}\Box A \vdash (\now{x} \semi Q) :: (z : C)}
  {\Omega, x{:}A \vdash Q :: (z : C)}
\end{array}
\]
  \vspace{-1em}
  \caption{Explicit Temporal Typing Rules}
  \vspace{-1em}
  \label{fig:temporal}
\end{figure}

\section{Preservation and Progress}
\label{sec:metatheory}


The main theorems that exhibit the deep connection between our type
system and the timed operational semantics are the usual \emph{type
  preservation} and \emph{progress}, sometimes called \emph{session
  fidelity} and \emph{deadlock freedom}, respectively.  Compared to
other recent treatments of linear session
types~\citep{Pfenning15fossacs,Balzer17icfp}, new challenges are
presented by abstract time and the temporal modalities.

\subsection{Configuration Typing}
\label{sec:config}

A key question is how we type configurations $\CC$.  Configurations
consist of multiple processes and messages, so they both \emph{use}
and \emph{provide} a collection of channels. And even though we
treat a configuration as a multiset, typing imposes a partial order on
the processes and messages where a provider of a channel appears to
the left of its client.
\[
\begin{array}{llcl}
\m{Configuration} & \CC & ::= & \cdot \mid \CC\; \CC' \mid \m{proc}(c, t, P)
\mid \m{msg}(c, t, M)
\end{array}
\]
We say $\m{proc}(c, t, P)$ and $\m{msg}(c, t, M)$ \emph{provide} $c$.
We stipulate that no two distinct processes or messages in a
configuration provide the same channel $c$.  Also recall that messages
$M$ are simply processes of a particular form and are typed as such.
We can read off the possible messages (of which there is one for each
type constructor) from the operational semantics.  They are summarized here
for completeness.
\[
\begin{array}{lcl}
M & ::= & (c.k \semi c \leftarrow c') \mid (c.k \semi c' \leftarrow c)
\mid \m{close}\; c \mid (\m{send}\; c\; d \semi c' \leftarrow c)
\mid (\m{send}\; c\; d \semi c \leftarrow c')
\end{array}
\]

The typing judgment has the form $\Omega' \vDash \CC :: \Omega$
meaning that if composed with a configuration that provides $\Omega'$,
the result will provide $\Omega$.
\[
\begin{array}{c}
\infer[\m{empty}]
  {\Omega \vDash (\cdot) :: \Omega}
  {\mathstrut}
\hspace{3em}
\infer[\m{compose}]
  {\Omega_0 \vDash (\CC_1\; \CC_2) :: \Omega_2}
  {\Omega_0 \vDash \CC_1 :: \Omega_1
  & \Omega_1 \vDash \CC_2 :: \Omega_2}
\end{array}
\]
To type processes and messages, we begin by
considering \emph{preservation}: we would like to achieve that if
$\Omega' \vDash \CC :: \Omega$ and $\CC \mapsto \CC'$ then still
$\Omega' \vDash \CC' :: \Omega$. Without the temporal modalities, this
is guaranteed by the design of the sequent calculus: the right and
left rules match just so that cut reduction (which is the basis for
reduction in the operational semantics) leads to a well-typed
deduction. The key here is what happens with time.  Consider
the special case
\[
\begin{array}{c}
\infer[{\Next}{LR}]
  {\Next \Omega \vdash (\delay \semi P) :: (x : \Next A)}
  {\Omega \vdash P :: A}
\qquad
\infer[\m{proc}(c, t, \delay \semi P) \;\mapsto\; \m{proc}(c, t+1, P)]
{\mathstrut}
{\mathstrut}
\end{array}
\]
Note that, inevitably, the type of the channel $c$ changes in the
transition, from $c : \Next A$ to $c : A$ and similarly for all
channels used by $P$. So if in $\m{proc}(c, t, Q)$ we were to use the
type of $Q$ as the type of the semantic process object, preservation
would fail.  But while the type changes from $\Next A$ to $A$,
\emph{time} also advances from $t$ to $t+1$. This suggests the
following rule should keep the configuration type invariant:
\[
\infer[\m{proc}^{\Next}]
  {\Next^t \Omega \vDash \m{proc}(c,t,P) :: (c : \Next^t A)}
  {\Omega \vdash P :: (c : A)}
\]
When we transition from $\delay \semi P$ to $P$ we strip one $\Next$
modality from $\Omega$ and $A$, but because we also advance time from
$t$ to $t+1$, the $\Next$ modality is restored,
keeping the interface type invariant.

When we also consider types $\Box A$ and $\Dia A$ the situation is a
little less straightforward because of their interaction with $\Next$,
as we have already encountered in Section~\ref{sec:interact}.  We
reuse the idea of the solution, allowing the subtraction of time from
a type, possibly stopping when we meet a $\Box$ or $\Dia$.
\[
\begin{array}{lcllcl}
[A]_L^{-0} & = & A & [A]_R^{-0} & = & A \\\relax
[A]_L^{-(t+1)} & = & [[A]_L^{-t}]_L^{-1} & [A]_R^{-(t+1)} & = & [[A]_R^{-t}]_R^{-1}
\end{array}
\]
This is extended to channel declarations in the obvious way.
Additionally, the imprecision of $\Box A$ and $\Dia A$ may create
temporal gaps in the configuration that need to be bridged by a weak
form of subtyping $A \wsubt B$ (not to be confused with the much
stronger form $A \leq B$ in Section~\ref{sec:subtyping}),
\[
\begin{array}{c}
\infer[\Box_{\m{weak}}]
  {\Next^m \Box A \wsubt \Next^n \Box A}
  {m \leq n}
\hspace{3em}
\infer[\Dia_{\m{weak}}]
  {\Next^m \Dia A \wsubt \Next^n \Dia A}
  {m \geq n}
\hspace{3em}
\infer[\m{refl}]
  {A \wsubt A}
  {\mathstrut}
\end{array}
\]
This relation is specified to be reflexive and clearly transitive.
We extend it to contexts $\Omega$ in the obvious manner.  In our final
rules we also account for some channels that are not used by $P$ or
$M$ but just passed through.
\[
\infer[\m{proc}]
  {\Omega_0, \Omega' \vDash \m{proc}(c,t,P) :: (\Omega_0, c : A')}
  {\Omega' \wsubt \Omega & [\Omega]_L^{-t} \vdash P :: [c : A]_R^{-t}
    & A \wsubt A'}
\hspace{1em}
\infer[\m{msg}]
  {\Omega_0, \Omega' \vDash \m{msg}(c,t,M) :: (\Omega_0, c : A')}
  {\Omega' \wsubt \Omega & [\Omega]_L^{-t} \vdash M :: [c : A]_R^{-t}
    & A \wsubt A'}
\]

\subsection{Type Preservation}
\label{sec:preservation}

With the four rules for typing configurations
($\m{empty}$, $\m{compose}$, $\m{proc}$ and $\m{msg}$),
type preservation is
relatively straightforward.  We need some standard lemmas about
being able to split a configuration and be able to move a provider
(whether process or message) to the right in a typing derivation until
it rests right next to its client.  Regarding time shifts, we need the
following properties.
\begin{lemma}[Time Shift]
\label{lm:shift}
\mbox{}
\begin{enumerate}
\item[(i)] If $[A]_L^{-t} = [B]_R^{-t}$ and both are defined
  then $A = B$.
\item[(ii)] $[[A]_L^{-t}]_L^{-s} = [A]_L^{-(t+s)}$ and if either side is
  defined, the other is as well.
\item[(iii)] $[[A]_R^{-t}]_R^{-s} = [A]_R^{-(t+s)}$ and if either side
  is defined, the other is as well.
\end{enumerate}
\end{lemma}
\begin{theorem}[Type Preservation]
\label{thm:preservation}
  If $\Omega' \vDash \CC :: \Omega$ and $\CC \mapsto \DD$ then $\Omega' \vDash \DD :: \Omega$.
\end{theorem}
\begin{proof}
  By case analysis on the transition rule, applying inversion to the
  given typing derivation, and then assembling a new derivation of
  $\DD$.
\end{proof}
Type preservation on basic session types is a simple special case of
this theorem.

\subsection{Global Progress}
\label{sec:progress}

We say a process or message is \emph{poised} if it is trying to
communicate along the channel that it provides.  A poised process is
comparable to a value in a sequential language. A configuration is
poised if every process or message in the configuration is poised.
Conceptually, this implies that the configuration is trying to communicate
externally, i.e. along one of the channel it provides.
The progress theorem then shows that either a configuration can take a
step or it is poised.  To prove this we show first that the typing
derivation can be rearranged to go strictly from right to left and
then proceed by induction over this particular derivation.  This much
is standard, even for significantly more complicated session-typed
languages~\citep{Balzer17icfp}.

The question is how can we prove that processes are either at the same
time (for most interactions) or that the message recipient is ready
before the message arrives (for $\whenn$, $\noww$, and some forwards)?
The key insight here is in the following lemma.

\begin{lemma}[Time Inversion]
\label{lm:time-inversion}
\mbox{}
\begin{enumerate}
\item[(i)] If $[A]_R^{-s} = [A]_L^{-t}$ and either side starts with a
  basic session type constructor then $s = t$.
\item[(ii)] If $[A]_L^{-t} = \Box B$ and $[A]_R^{-s} \neq \Next(-)$
  then $s \leq t$ and $[A]_R^{-s} = \Box B$.
\item[(iii)] If $[A]_R^{-t} = \Dia B$ and $[A]_L^{-s} \neq \Next(-)$
  then $s \leq t$ and $[A]_L^{-s} = \Dia B$.
\end{enumerate}
\end{lemma}

\begin{theorem}[Global Progress]
\label{thm:progress}
\mbox{}
If $\cdot \vDash \CC :: \Omega$ then either
\begin{enumerate}
\item[(i)] $\CC \mapsto \CC'$ for some $\CC'$, or
\item[(ii)] $\CC$ is poised.
\end{enumerate}
\end{theorem}
\begin{proof}
  By induction on the right-to-left typing of $\CC$ so that either
  $\CC$ is empty (and therefore poised) or
  $\CC = (\DD\; \m{proc}(c, t, P))$ or
  $\CC = (\DD\; \m{msg}(c, t, M)$.  By induction hypothesis, $\DD$ can
  either take a step (and then so can $\CC$), or $\DD$ is poised.  In
  the latter case, we
  analyze the cases for $P$ and $M$, applying multiple steps of
  inversion and Lemma~\ref{lm:time-inversion} to show that in each
  case either $\CC$ can take a step or is poised.
\end{proof}


\section{Time Reconstruction}
\label{sec:recon}


The process expressions introduced so far have straightforward
syntax-directed typing rules. This requires the programmer to write a
significant number of explicit $\delay$, $\whenn$, and $\noww$
constructs in their code. This in turn hampers reuse: we would like to
be able to provide multiple types for the same process definition so
it can be used in different contexts, with different types, even under
a single, fixed cost model.

In this section we introduce an implicit system which may be thought
of as a \emph{temporal refinement} of the basic session type system in
Section~\ref{sec:basic}.  The $\delay$, $\whenn$, and $\noww$
constructs never appear in the source, and, as before, $\tick$ is
added before type-checking and never by the programmer.  The rules for
the new judgment $\Omega \vdashi P :: (x : A)$ are shown in
Figure~\ref{fig:implicit}; the other rules remain the same (except for
$\m{def}$, see below). We still need an explicit rule for the $\tick$
synonym of $\delay$ which captures the cost model.

\begin{figure}
\centering
\[
\begin{array}{c}
\infer[{\Next}{LR}]
  {\Omega \vdashi P :: (x : A)}
  {[\Omega]_L^{-1} \vdashi P :: [x : A]_R^{-1}}
\hspace{3em}
\infer[{\Next}{LR}']
  {\Omega \vdashi (\tick \semi P) :: (x : A)}
  {[\Omega]_L^{-1} \vdashi P :: [x : A]_R^{-1}}
\\[1em]
\infer[{\Dia}R]
  {\Omega \vdashi P :: (x : \Dia A)}
  {\Omega \vdashi P :: (x : A)}
\hspace{3em}
\infer[{\Dia}L]
  {\Omega, x{:}\Dia A \vdashi Q :: (z : C)}
  {\Omega\; \m{delayed}^\Box
    & \Omega, x{:} A \vdashi Q :: (z : C)
    & C\; \m{delayed}^\Dia}
\\[1em]
\infer[{\Box}R]
  {\Omega \vdashi P :: (x : \Box A)}
  {\Omega\; \m{delayed}^\Box
    & \Omega \vdashi P :: (x : A)}
\hspace{3em}
\infer[{\Box}L]
  {\Omega, x{:}\Box A \vdashi Q :: (z : C)}
  {\Omega, x{:}A \vdashi Q :: (z : C)}
\end{array}
\]
\vspace{-1em}
\caption{Implicit Temporal Rules}
\vspace{-1em}
\label{fig:implicit}
\end{figure}

\begin{figure}
\centering
\[
\begin{array}{c}
\infer[\m{refl}]
  {A \leq A}
  {\mathstrut}
\hspace{3em}
\infer[{\Next}{\Next}]
  {\Next A \leq \Next B}
  {A \leq B}
\hspace{2em}
\infer[{\Box}{\Next}]
  {\Box A \leq \Next B}
  {\Box A \leq B}
\hspace{2em}
\infer[{\Next}{\Dia}]
  {\Next A \leq \Dia B}
  {A \leq \Dia B}
\\[1em]
\infer[{\Box}R]
  {\Next^n \Box A \leq \Box B}
  {\Next^n \Box A \leq B}
\hspace{3em}
\infer[{\Box}L]
  {\Box A \leq B}
  {A \leq B}
\hspace{3em}
\infer[{\Dia}R]
  {A \leq \Dia B}
  {A \leq B}
\hspace{3em}
\infer[{\Dia}L]
  {\Dia A \leq \Next^n \Dia B}
  {A \leq \Next^n \Dia B}
\end{array}
\]
\vspace{-1em}
\caption{Subtyping Rules}
\vspace{-1em}
\label{fig:subtyping}
\end{figure}

These rules are trivially sound and complete with respect to the
explicit system in Section~\ref{sec:temporal} because from an implicit
type derivation we can read off the explicit process expression and
vice versa.  They are also manifestly decidable because the types in
the premises are smaller than those in the conclusion, with one
possible exception: In the ${\Next}LR$ rule the premise may be
equal to the conclusion if neither $\Omega$ nor $A$ contain a
type of the form $\Next(-)$.  In this case, $B = \Box B'$ for every
$y : B$ in $\Omega$ and $A = \Dia A'$ and there $P$ can delay by any
finite number of time steps.  Time reconstruction avoids such an
arbitrary delay.

Our examples revealed a significant shortcoming in these
rules: when calling upon a process definition, the types in the
antecedent and succedent often do not match the types of the process
to be spawned.  For example, the process $\mi{skip1s}$ in
Section~\ref{sec:eventually} we have
\begin{sill}
$\m{bits} = {\oplus}\{\m{b0} : \Next \m{bits}, \m{b1} : \Next \m{bits}, \m{\$} : \Next \one\}$ \\
$\m{sbits} = {\oplus}\{\m{b0} : \Next \m{sbits}, \m{b1} : \Next \Dia \m{sbits}, \m{\$} : \Next \one\}$ \\[1ex]
$y : \m{bits} \vdash \mi{compress} :: (x : \Next \m{sbits})$ \\
$y : \m{bits} \vdash \mi{skip1s} :: (x : \Next \Dia \m{sbits})$ \\[1ex]
$x \leftarrow \mi{skip1s} \leftarrow y =$ \\
\quad $\m{case}\; y\;$ \= $(\,\m{b1} \Rightarrow $ \= $\tick \semi $ \hspace{6em}\=\% $y : \m{bits} \vdash x : \Dia \m{sbits}$ \\
\>\> $x \leftarrow \mi{skip1s} \leftarrow y$ \>\% \emph{does not type-check!} \\
\> $\mid \ldots \,)$
\end{sill}
The indicated line does not type-check (neither in the explicit nor
the implicit system presented so far) since the type
$\Next \Dia \m{sbits}$ offered by $\mi{skip1s}$ does not match
$\Dia \m{sbits}$. We had to write a process $\mi{idle}$ to account
for this mismatch:
\begin{sill}
$x':\Next \Dia A \vdash \mi{idle} :: (x : \Dia A)$ \\
$x \leftarrow \mi{idle} \leftarrow x' = \delay \semi x \leftarrow x'$
\end{sill}
In the implicit system the version with an explicit identity
\emph{can} in fact be reconstructed:
\begin{sill}
$x \leftarrow \mi{skip1s} \leftarrow y =$ \\
\quad $\m{case}\; y\;$ \= $(\,\m{b1} \Rightarrow $ \= $\tick \semi $ \hspace{6em}\=\% $y : \m{bits} \vdashi x : \Dia \m{sbits}$ \\
\>\> $x' \leftarrow \mi{skip1s} \leftarrow y$ \>\% $x' : \Next \Dia \m{sbits} \vdashi x : \Dia \m{sbits}$ \\
\>\>\>\% $x' : \Dia \m{sbits} \vdashi x : \Dia \m{sbits}$ \quad using rule ${\Next}{LR}$ \\
\>\> $x \leftarrow x'$ \\
\> $\mid \ldots \,)$
\end{sill}

\subsection{Subtyping}
\label{sec:subtyping}

Extrapolating from the example of $\mi{skip1s}$ above, we can
generalize process invocations by allowing \emph{subtyping} on all
used channels.  The implicit rule for process invocation then reads
\[
\begin{array}{c}
\infer[\m{def}]
  {\Omega, \Omega' \vdashi (x \leftarrow f \leftarrow \Omega' \semi Q) :: (z : C)}
  {\Omega' \leq \Omega_f & (\Omega_f \vdashi f = P_f :: (x : A)) \in \Sigma
    & \Omega, x{:}A \vdashi Q :: (z : C)}
\end{array}
\]
But how do we define subtyping $A \leq B$? We would like the coercion
to be an identity on basic session types and just deal with temporal
mismatches through appropriate $\delay$, $\whenn$, and $\noww$
actions.  In other words, $A$ should be a subtype of $B$ if and only
if $y : A \vdashi x \leftarrow y :: (x : B)$.  Given this desired
theorem, we can just read off the subtyping rules from the
implicit typing rules in Figure~\ref{fig:implicit} by using the
forwarding process $x \leftarrow y$ as the subject in each rule!  This
form of subtyping is independent from subtyping between basic session
types~\citep{Gay05acta}, which we believe can be added to our system
in a sound way, even if it would not be complete for asynchronous
communication~\citep{Lange17fossacs}.

This approach yields the rules in Figure~\ref{fig:subtyping}, where we have split
the ${\Next}LR$ rule into three different cases.  We have expanded the
definitions of \emph{patient} types to make it syntactically more
self-contained.

\begin{theorem}[Subtyping Identity]
\label{thm:subid}
\mbox{}
$A \leq B$ iff $y : A \vdashi x \leftarrow y :: (x : B)$
\end{theorem}
\begin{proof}
  In each direction by straightforward induction over the structure of
  the given deduction.
\end{proof}

The subtyping rules are manifestly decidable.
In the bottom-up search for a subtyping derivation, the rules
${\Next}{\Next}$, ${\Box}R$, and ${\Dia}L$ can be applied eagerly
without losing completeness.  There is a nontrivial decision point
between the ${\Box}{\Next}$ and ${\Box}L$ rules.  The examples
$\Box S \leq \Next \Box S$ and $\Box \Next S \leq \Next S$ for a basic
session type $S$ show that sometimes ${\Box}{\Next}$ must be chosen
and sometimes ${\Box}L$ when both rules apply.  A dual
nondeterministic choice exists between ${\Next}{\Dia}$ and ${\Dia}R$.
The cost of backtracking is minimal in all examples we have
considered.

We already know that subtype coercions are identities.  To verify that
we have a sensible subtype relation it remains to prove that
transitivity is admissible. For this purpose we need two lemmas
regarding patient types, as they appear in the ${\Box}R$ and ${\Dia}L$
rules.

\begin{lemma}[Patience]
\label{lm:patience}
\mbox{}
\begin{enumerate}
\item[(i)] If $A \leq \Next^n \Box B$ then $A = \Next^k \Box A'$ for some $k$ and $A'$.
\item[(ii)] If $\Next^n \Dia A \leq B$ then $B = \Next^k \Dia B'$ for some $k$ and $B'$.
\end{enumerate}
\end{lemma}
\begin{proof}
  By separate inductions over the structure of the given deductions.
\end{proof}

\begin{lemma}[Impatience]
\label{lm:impatience}
\mbox{}
\begin{enumerate}
\item[(i)] If $\Next \Next^n \Box A \leq B$ then $\Next^n \Box A \leq B$.
\item[(ii)] If $A \leq \Next \Next^n \Dia B$ then $A \leq \Next^n \Dia B$.
\end{enumerate}
\end{lemma}
\begin{proof}
  By separate inductions over the structure of the given deductions.
\end{proof}

\begin{theorem}[Transitivity of Subtyping]
\label{lm:trans}
\mbox{}\newline
If $A \leq B$ and $B \leq C$ then $A \leq C$.
\end{theorem}
\begin{proof}
  By simultaneous induction on the structure of the deductions
  $\DD$ of $A \leq B$ and $\EE$ of $B \leq C$ with appeals
  to the preceding lemmas in four cases.
\end{proof}

\section{Further Examples}
\label{sec:examples}


In this section we present example analyses of some of the properties
that we can express in the type system, such as the message rates of
streams, the response time of concurrent data structures, and the span
of a fork/join parallel program. 

In some examples we use parametric definitions, both at the level of
types and processes.  For example, $\m{stack}_A$
describes stacks parameterized over a type $A$, $\m{list}_A[n]$
describes lists of $n$ elements, and $\m{tree}[h]$ describes binary
trees of height $h$.  Process definitions are similarly
parameterized. We think of these as families of ordinary definitions
and calculate with them accordingly, at the metalevel, which is
justified since they are only implicitly quantified across whole
definitions. This common practice (for example, in work on interaction
nets~\citet{GimenezPOPL16}) avoids significant syntactic overhead,
highlighting conceptual insight. It is of course possible to
internalize such parameters (see, for example, work on refinement of
session types~\citep{Griffith13nasa} or explicitly polymorphic session
types~\citep{Caires13esop,Griffith16phd}).

\subsection{Response Times: Stacks and Queues}

To analyze response times, we study concurrent stacks and queues.
A stack data structure provides a client with a choice between a push
and a pop.  After a push, the client has to send an element, and the
provider will again behave like a stack.  After a pop, the provider
will reply either with the label $\m{none}$ and terminate (if there
are no elements in the stack), or send an element and behave
again like a stack.  In the cost-free model, this is expressed
in the following session type.
\begin{sill}
  $\m{stack}_{A} = {\with}\{$ \= $\m{push} : A \lolli \m{stack}_A,$ \\
  \> $\m{pop} : {\oplus}\{$ \= $\m{none} : \one,$ \\
  \>\> $\m{some} : A \tensor \m{stack}_A\,\}\,\}$
\end{sill}
We implement a stack as a chain of processes.  The bottom to the
stack is defined by the process $\mi{empty}$, while a process
$\mi{elem}$ holds a top element of the stack as well as a channel with
access to the top of the remainder of the stack.
\begin{sill}
  $x : A, t : \m{stack}_A \vdash \mi{elem} :: (s : \m{stack}_A)$ \\
  $\cdot \vdash \mi{empty} :: (s : \m{stack}_A)$
\end{sill}

The cost model we would like to consider here is $\mc{RS}$ where both
receives and sends cost one unit of time.  Because a receive costs one
unit, every continuation type must be delayed by one tick of the
clock, which we have denoted by prefixing continuations by the $\Next$
modality. This delay is not an artifact of the implemention, but an
inevitable part of the cost model---one reason we have distinguished
the synonyms $\tick$ (delay of one, due to the cost model) and
$\delay$ (delay of one, to correctly time the interactions).  In this
section of examples we will make the same distinction for the
next-time modality: we write $\tock A$ for a step in time mandated by
the cost model, and $\Next A$ for a delay necessitated by a
particular set of process definitions.

As a first approximation, we would have
\begin{sill}
  $\m{stack}_{A} = {\with}\{$ \= $\m{push} : \tock (A \lolli \tock \m{stack}_A),$ \\
  \> $\m{pop} : \tock {\oplus}\{$ \= $\m{none} : \tock \one,$ \\
  \>\> $\m{some} : \tock (A \tensor \tock \m{stack}_A)\,\}\,\}$
\end{sill}
There are several problems with this type.  The stack is a data
structure and has little or no control over \emph{when} elements will
be pushed onto or popped from the stack.  Therefore we should use a
type $\Box \m{stack}_A$ to indicate that the client can choose the
times of interaction with the stack.  While the elements are held by
the stack time advances in an indeterminate manner.  Therefore, the
elements stored in the stack must also have type $\Box A$, not $A$
(so that they are always available).
\begin{sill}
  $\m{stack}_{A} = {\with}\{$ \= $\m{push} : \tock (\Box A \lolli \tock \Box \m{stack}_A),$ \\
  \> $\m{pop} : \tock {\oplus}\{$ \= $\m{none} : \tock \one,$ \\
  \>\> $\m{some} : \tock (\Box A \tensor \tock \Box \m{stack}_A)\,\}\,\}$ \\[1ex]
  $x : \Box A, t : \Box \m{stack}_A \vdash \mi{elem} :: (s : \Box \m{stack}_A)$ \\
  $\cdot \vdash \mi{empty} :: (s : \Box \m{stack}_A)$
\end{sill}
This type expresses that the data structure is very efficient in its
response time: there is no additional delay after it receives a
$\m{push}$ and then an element of type $\Box A$ before it can take the
next request, and it will respond immediately to a $\m{pop}$ request.
It may not be immediately obvious that such an efficient
implementation actually exists in the $\mc{RS}$ cost model, but it
does.  We use the implicit form from Section~\ref{sec:recon} omitting
the $\tick$ constructs after each receive and send, and also the
$\whenn$ before each $\m{case}$ that goes along with type $\Box A$.
\begin{sill}
  $s \leftarrow \mi{elem} \leftarrow x\; t =$ \\
  \quad \= $\m{case}\; s\;$ \= $(\, \m{push} \Rightarrow$ \= $y \leftarrow \m{recv}\; s \semi$ \\
  \>\>\> $s' \leftarrow \mi{elem} \leftarrow x\; t \semi$ \hspace{3em} \=\% \textit{previous top of stack, holding $x$} \\
  \>\>\> $s \leftarrow \mi{elem} \leftarrow y\; s'$ \>\% \textit{new top of stack, holding $y$} \\
  \>\> $\mid \m{pop} \Rightarrow$ \> $s.\m{some} \semi$ \\
  \>\>\> $\m{send}\; s\; x \semi$  \>\% \textit{send channel $x$ along $s$} \\
  \>\>\> $s \leftarrow t\,)$ \>\% \textit{$s$ is now provided by $t$, via forwarding} \\[1em]
  $s \leftarrow \mi{empty} =$ \\
  \> $\m{case}\; s\;$ \= $(\, \m{push} \Rightarrow$ \= $y \leftarrow \m{recv}\; s \semi$ \\
  \>\>\> $e \leftarrow \mi{empty} \semi$ \>\% \textit{new bottom of stack} \\
  \>\>\> $s \leftarrow \mi{elem} \leftarrow y\; e$ \\
  \>\> $\mid \m{pop} \Rightarrow$ \> $s.\m{none} \semi$ \\
  \>\>\> $\m{close}\; s\,)$
\end{sill}

The specification and implementation of a queue is very similar.
The key difference in the implementation is that when we receive
a new element we pass it along the chain of processes until it
reaches the end.  So instead of  
\begin{sill}
  $s' \leftarrow \mi{elem} \leftarrow x\; t \semi$ \hspace{3em} \=\% \textit{previous top of stack, holding $x$} \\
  $s \leftarrow \mi{elem} \leftarrow y\; s'$ \>\% \textit{new top of stack, holding $y$}
\end{sill}
we write
\begin{sill}
  $s' \leftarrow \mi{elem} \leftarrow x\; t \semi$ \hspace{3em} \=\% \textit{previous top of stack, holding $x$} \kill
  $t.\m{enq} \semi$ \\
  $\m{send}\; t\; y \semi$ \>\% \textit{send $y$ to the back of the queue} \\
  $s \leftarrow \mi{elem} \leftarrow x\; t$
\end{sill}
These two send operations take two units of time, which must be
reflected in the type: after a channel of type $\Box A$ has been
received, there is a delay of an additional two units of time before
the provider can accept the next request.
\begin{sill}
  $\m{queue}_{A} = {\with}\{$ \= $\m{enq} : \tock (\Box A \lolli \tock \Next \Next \Box \m{queue}_A),$ \\
  \> $\m{deq} : \tock {\oplus}\{$ \= $\m{none} : \tock \one,$ \\
  \>\> $\m{some} : \tock (\Box A \tensor \tock \Box \m{queue}_A)\,\}\,\}$ \\[1ex]
  $x : \Box A, t : \Next \Next \Box \m{queue}_A \vdash \mi{elem} :: (s : \Box \m{queue}_A)$ \\
  $\cdot \vdash \mi{empty} :: (s : \Box \m{queue}_A)$
\end{sill}
Time reconstruction will insert the additional delays in the
$\mi{empty}$ process through subtyping, using
$\Box \m{queue}_A \leq \Next \Next \Box \m{queue}_A$.  We have
syntactically expanded the tail call so the second use of subtyping is
more apparent.
\begin{sill}
  $s \leftarrow \mi{empty} =$ \\
  \quad \= $\m{case}\; s\;$ \= $(\, \m{enq} \Rightarrow$ \= $y \leftarrow \m{recv}\; s \semi$
\hspace{5em} \=\% $y : \Box A \vdash s : \Next \Next \Box \m{queue}_A$ \\
  \>\>\> $e \leftarrow \mi{empty} \semi$ \>\% $y : \Box A, e : \Box \m{queue}_A \vdash s : \Next \Next \Box \m{queue}_A$ \\
  \>\>\> $s' \leftarrow \mi{elem} \leftarrow y\; e \semi$ \>\% $\Box \m{queue}_A \leq \Next \Next \Box \m{queue}_A$ (on $e$) \\
  \>\>\> $s \leftarrow s'$ \>\% $\Box \m{queue}_A \leq \Next \Next \Box \m{queue}_A$ (on $s'$) \\
  \>\> $\mid \m{deq} \Rightarrow$ \> $s.\m{none} \semi$ \\
  \>\>\> $\m{close}\; s\,)$
\end{sill}
The difference between the \emph{response times} of stacks and queues
in the cost model is minimal: both are constant, with the queue being
two units slower.  This is in contrast to the total
work~\cite{Das18arxiv} which is constant for the stack but linear in
the number of elements for the queue.

This difference in response times can be realized by typing clients of
both stacks and queues. We compare clients $S_n$ and $Q_n$ that
insert $n$ elements into a stack and queue, respectively, send the
result along channel $d$, and then terminate.  We show only their type
below, omitting the implementations.
\begin{sill}
 $x_1 : \Box A, \ldots, x_n : \Box A, s : \Box \m{stack}_A \vdash S_n ::
 (d : \Next^{2n}\, (\Box \m{stack}_A \tensor \tock \one))$\\
 $x_1 : \Box A, \ldots, x_n : \Box A, s : \Box \m{queue}_A \vdash Q_n ::
 (d : \Next^{4n}\, (\Box \m{queue}_A \tensor \tock \one))$
\end{sill}
The types demonstrate that the total execution time of $S_n$ is only $2n+2$,
while it is $4n+2$ for $Q_n$. The difference comes from the difference in
response times. Note that we can infer precise execution times, even
in the presence of the $\Box$ modality in the stack and queue types.

\subsection{Parametric Rates: Lists and Streams}

Lists describe an interface that sends either $\m{nil}$ and ends the
session, or sends $\m{cons}$ followed by a channel of some type $A$ and
then behaves again like a list.  In the cost-free setting:
\begin{sill}
  $\m{list}_A = {\oplus}\{\, \m{cons} : A \tensor \m{list}_A, \m{nil} : \one\,\}$
\end{sill}
Here is the straightforward definition of $\mi{append}$.
\begin{sill}
  $l_1 : \m{list}_A, l_2 : \m{list}_A \vdash \mi{append} : (l : \m{list}_A)$ \\[1ex]
  $l \leftarrow \mi{append} \leftarrow l_1\; l_2 =$ \\
  \quad \= $\m{case}\; l_1\;$ \= $(\, \m{cons} \Rightarrow$ \= $x \leftarrow \m{recv}\; l_1 \semi$ \hspace{5em} \= \% receive element $x$ from $l_1$ \\
  \>\>\> $l.\m{cons} \semi \m{send}\; l\; x \semi$ \> \% send $x$ along $l$\\
  \>\>\> $l \leftarrow \mi{append} \leftarrow l_1\; l_2$ \> \% recurse \\
  \>\> $\mid \m{nil} \Rightarrow$ \> $\m{wait}\; l_1 \semi$ \> \% wait for
  $l_1$ to close\\
  \>\>\> $l \leftarrow l_2\,)$ \> \% identify $l$ and $l_2$
\end{sill}
In this example we are interested in analyzing the timing of several
processes precisely, but parametrically over an arrival rate.  Because
it takes two units of time to copy the inputs to the outputs, the
arrival rate needs to be at least 2, which we represent by writing it
as $r+2$. Since we append the two lists, the second list will be idle
while we copy the elements from the first list to the output. We could
give this list type $\Box(-)$, but we can also precisely determine the
delay if we index lists by the number of elements.  We write
$\m{list}_A[n]$ for a list sending exactly $n$ elements.  We have the
following types in the $\mc{RS}$ cost model:
\begin{sill}
  $\m{list}_A[0] = {\oplus}\{$ \= $\m{nil} : \tock \one \,\}$ \\
  $\m{list}_A[n+1] = {\oplus}\{$ \= $\m{cons} : \tock (\Box A \tensor \tock \Next^{r+2}\, \m{list}_A[n])\,\}$
\end{sill}
As before, the tick marks account for the delay mandated by the cost
model.  The $\Next^{r+2}$ accounts for the arrival rate of $r+2$. We
use type $\Box A$ for the elements since they will be in the lists for
an indeterminate amount of time.  The precise type of $\mi{append}$
then becomes
\begin{sill}
  $l_1 : \m{list}_A[n], l_2 : \Next^{(r+4)n+2}\, \m{list}_A[k]
  \vdash \mi{append} :: (l : \Next \Next \m{list}_A[n+k])$
\end{sill}
It expresses that the output list has the same rate as the input
lists, but with a delay of 2 cycles relative to $l_1$.  The channel
$l_2$ has to sit idle for $r+4$ cycles for each element of $l_1$,
accounting for the two inputs along $l_1$ and two outputs along $l_2$.
It takes 2 further cycles to input the $\m{nil}$ and the end token for
the list.

With our type system and just a little bit of arithmetic we can verify
this type, checking the definition twice: once for a list of length
$0$ and once for $n+1$.  We show here the latter, where
$l_1 : \m{list}_A[n+1]$.
\begin{tabbing}
  $l \leftarrow \mi{append} \leftarrow l_1\; l_2 =$ \\
  \= $\m{case}\; l_1\;$ \= $(\, \m{cons} \Rightarrow$ \hspace{2em}\=\% 
$l_1{:}\Box A \tensor \tock \Next^{r+2}\, \m{list}_A[n],
l_2 : [\Next^{(r+4)(n+1)+2}\,\m{list}_A[k]]_L^{-1} \vdash l : \Next \m{list}_A[(n+1)+k]$ \\
  \>\> $x \leftarrow \m{recv}\; l_1 \semi$ \>\% 
$x{:}\Box A, l_1{:}\Next^{r+2}\, \m{list}_A[n],
l_2 : [\Next^{(r+4)(n+1)+2}\,\m{list}_A[k]]_L^{-2} \vdash l : \m{list}_A[(n+1)+k]$ \\
  \>\> $l.\m{cons} \semi$ \>\%
$x{:}\Box A, l_1{:}\Next^{r+1}\, \m{list}_A[n],
l_2 : [\Next^{(r+4)(n+1)+2}\,\m{list}_A[k]]_L^{-3} \vdash l : \Box A \tensor \tock
\Next^{r+2}\, \m{list}_A[n+k]$ \\
  \>\> $\m{send}\; l\; x \semi$ \>\%
$l_1{:}\Next^r\, \m{list}_A[n],
l_2 : [\Next^{(r+4)(n+1)+2}\,\m{list}_A[k]]_L^{-4} \vdash l : \Next^{r+2}\, \m{list}_A[n+k]$ 
\\
  \>\> \% $\delay^{\red{r}}$ \>\%
$l_1{:}\m{list}_A[n], 
l_2 : [\Next^{(r+4)(n+1)+2}\,\m{list}_A[k]]^{-4-r} \vdash l : \Next^{2}\, \m{list}_A[n+k]$ \\
  \>\>\>\%
$l_1{:}\m{list}_A[n], 
l_2 : \Next^{(r+4)n+2}\,\m{list}_A[k] \vdash l : \Next \Next \m{list}_A[n+k]$ \\
  \>\> $l \leftarrow \mi{append} \leftarrow l_1\; l_2$ \\
  \>\> $\mid \m{nil} \Rightarrow \ldots\,)$
\end{tabbing}
We showed only the one delay by $r$ units inserted by time
reconstruction since it is the critical step. The case for $\m{nil}$
does not apply for $l_1 : \m{list}_A[n+1]$.  Here is the typing
derivation when $l_1 : \m{list}_A[0]$ where the $\m{cons}$ branch does
not apply.
\begin{sill}
  $l \leftarrow \mi{append} \leftarrow l_1\; l_2 = $\\
  \quad \= $\m{case}\;l_1\;$ \= $(\,\m{cons} \Rightarrow \ldots$ \\
  \>\> $\mid \m{nil} \Rightarrow $ \= \hspace{5em}\=\%
$l_1 : \one, l_2 : \Next \m{list}_A[k] \vdash l : \Next\m{list}_A[k]$ \\
  \>\>\> $\m{wait}\; l_1 \semi$ \>\% $l_2 : \m{list}_A[k] \vdash l : \m{list}_A[k]$\\
  \>\>\> $l \leftarrow l_2\,)$
\end{sill}

As a related example we consider a process that alternates the
elements between two infinite input streams. At first we might expect
if the two input streams come in with a rate of 2 then the output
stream will have a rate of 1.  However, in the $\mc{RS}$ cost model
one additional tick is required for sending on the messages which
means that the input streams need to have rate 3 and be offset by 2
cycles.  We parameterize the type of stream by its rate $k$
\begin{sill}
  $\m{stream}^k_A = \Box A \tensor \tock \Next^k\, \m{stream}_A^k$ \\[1ex]
  $l_1 : \m{stream}_A^3, l_2 : \Next^2\, \m{stream}_A^3 \vdash \mi{alternate} :: (l : \Next^1\, \m{stream}_A^1)$ \\[1ex]
  $l \leftarrow \mi{alternate} \leftarrow l_1\; l_2 =$ \\
  \quad \= $x \leftarrow \m{recv}\; l_1 \semi$ \hspace{5em}\=\% $x : \Box A,$ \= $l_1 : \Next^3\, \m{stream}_A^3, l_2 : \Next^1\, \m{stream}_A^3 \vdash l : \m{stream}_A^1$ \\
  \> $\m{send}\; l\; x \semi$ \>\% \> $l_1 : \Next^2\, \m{stream}_A^3, l_2 : \m{stream}_A^3 \vdash l : \Next^1\, \m{stream}_A^1$ \\
  \> $l \leftarrow \mi{alternate} \leftarrow l_2\; l_1\,)$
\end{sill}
A more general parametric type for the same code would be
\begin{sill}
  $l_1 : \m{stream}_A^{2k+3}, l_2 : \Next^{k+2}\, \m{stream}_A^{2k+3} \vdash \mi{alternate} :: (l : \Next^1\, \m{stream}_A^{k+1})$
\end{sill}
from which we can recover the more specialized one with $k = 0$.

\subsection{Span Analysis: Trees}

We use trees to illustrate an example that is typical for fork/join
parallelism and computation of \emph{span}.  In order to avoid
integers, we just compute the parity of a binary tree of height $h$
with boolean values at the leaves.  We do not show the obvious
definition of $\mi{xor}$, which in the $\mc{RS}$ cost model requires a
delay of four from the first input.
\begin{sill}
  $\m{bool} = {\oplus}\{\,\m{b0} : \tock \one, \m{b1} : \tock \one\,\}$ \\[1ex]
  $a : \m{bool}, b: \Next^2\, \m{bool} \vdash \mi{xor} :: (c : \Next^4\, \m{bool})$
\end{sill}
In the definition of $\mi{leaf}$ and $\mi{node}$ we have explicated
the delays inferred by time reconstruction, but not the $\tick$
delays.  The type of $\m{tree}[h]$ gives the \emph{span} of this
particular parallel computation as $5h+2$.  This is the time it takes
to compute the parity under maximal parallelism, assuming that
$\mi{xor}$ takes 4 cycles as shown in the type above.
\begin{sill}
  $\m{tree}[h] = {\with}\{\,\m{parity} : \tock \Next^{5h+2}\, \m{bool}\,\}$ \\[1ex]
  $\cdot \vdash \mi{leaf} :: (t : \m{tree}[h])$ \\[1ex]
  $t \leftarrow \mi{leaf} =$ \\
  \quad \= $\m{case}\; t$ \= $(\,\m{parity} \Rightarrow$ \=\hspace{7em}\=\% $\cdot \vdash t : \Next^{5h+2}\, \m{bool}$ \\
  \>\>\> \% $\delay^{\red{5h+2}}$ \>\% $\cdot \vdash t : \m{bool}$ \\
  \>\>\> $t.\m{b0} \semi$ \>\% $\cdot \vdash t : \one$ \\
  \>\>\> $\m{close}\; t\,)$ \\[1em]
  $l : \Next^1 \m{tree}[h], r : \Next^3\, \m{tree}[h] \vdash \mi{node} :: (t : \mi{tree}[h+1])$ \\[1ex]
  $t \leftarrow \mi{node} \leftarrow l\; r =$ \\
  \> $\m{case}\;t$ \> $(\,\m{parity} \Rightarrow$ \>\>\% $l : \m{tree}[h], r: \Next^2\, \m{tree}[h] \vdash t : \Next^{5(h+1)+2}\, \m{bool}$ \\
  \>\>\> $l.\m{parity} \semi$ \>\% $l : \Next^{5h+2}\, \m{bool}, r : \Next^1 \m{tree}[h] \vdash t : \Next^{5(h+1)+1}\, \m{bool}$ \\
  \>\>\> \% $\delay$          \>\% $l : \Next^{5h+1}\, \m{bool}, r : \m{tree}[h] \vdash t : \Next^{5h+5}\, \m{bool}$ \\
  \>\>\> $r.\m{parity} \semi$ \>\% $l : \Next^{5h}\, \m{bool}, r: \Next^{5h+2}\, \m{bool} \vdash t : \Next^{5h+4}\, \m{bool}$\\
  \>\>\> \% $\delay^{\red{5h}}$ \>\% $l : \m{bool}, r:\Next^2\, \m{bool} \vdash t : \Next^4\, \m{bool}$ \\
  \>\>\> $t \leftarrow \mi{xor} \leftarrow l\; r\,)$
\end{sill}
The type $l : \Next^1\, \m{tree}[h]$ comes from the fact that, after
receiving a $\m{parity}$ request, we first send out the $\m{parity}$
request to the left subtree $l$.  The type $r : \Next^3\, \m{tree}[h]$
is determined from the delay of 2 between the two inputs to
$\mi{xor}$.  The magic number 5 in the type of $\m{tree}$ was derived
in reverse from setting up the goal of type-checking the $\mi{node}$
process under the constraints already mentioned.  We can also think of
it as 4+1, where 4 is the time to compute the exclusive or at each
level and 1 as the time to propagate the $\m{parity}$ request down each
level.

As is often done in abstract complexity analysis, we can also
impose an alternative cost model.  For example, we may count
only the number of calls to $\mi{xor}$ while all other operations
are cost free.  Then we would have
\begin{sill}
  $a:\m{bool}, b:\m{bool} \vdash \mi{xor} :: (c:\Next \m{bool})$ \\
  $\m{tree}[h] = {\with}\{\, \m{parity} : \Next^h\, \m{bool}\, \}$ \\[1em]
  $\cdot \vdash \m{leaf} :: (t : \m{tree}[h])$ \\
  $l : \m{tree}[h], r : \m{tree}[h] \vdash \m{node} :: (t : \m{tree}[h+1])$
\end{sill}
with the same code but different times and delays from before.  The
reader is invited to reconstruct the details.

\subsection{A Higher-Order Example}

As an example of higher-order programming we show how to encode a
process analogue of a fold function.  Because our language is
purely linear the process to fold over a list has to be recursively
defined.  In the cost-free setting we would write
\begin{sill}
  $\m{folder}_{AB} = {\with}\{$ \= $\m{next} : A \lolli (B \lolli (B \tensor \m{folder}_{AB})),
  \m{done} : \one\,\}$ \\[1ex]
  $l : \m{list}_A, f : \m{folder}_{AB}, b : B \vdash \mi{fold} :: (r : B)$ \\[1ex]
  $r \leftarrow \mi{fold} \leftarrow l\; f\; b =$ \\
  \quad \= $\m{case}\; l\;$ \= $(\,\m{cons} \Rightarrow$ \= $x \leftarrow \m{recv}\; l \semi$ \\
  \>\>\> $f.\m{next} \semi \m{send}\; f\; x \semi \m{send}\; f\; b \semi$
  \hspace{2em} \= \% send $x$ and $b$ to folder $f$ \\
  \>\>\> $y \leftarrow \m{recv}\; f \semi$
  $r \leftarrow \mi{fold} \leftarrow l\; f\; y$
  \> \% receive $y$ from $f$ and recurse \\
  \>\> $\mid \m{nil} \Rightarrow$ \> $\m{wait}\; l \semi$ 
  $f.\m{done} \semi \m{wait}\; f \semi r \leftarrow b\,)$
\end{sill}
If we want to assign precise temporal types to the $\mi{fold}$ process
then the incoming list should have a delay of at least 4 between
successive elements. Working backwards from the code we obtain
the following types.
\begin{sill}
  $\m{list}_A[0] = {\oplus}\{\,\m{nil} : \tock \one\,\}$ \\
  $\m{list}_A[n+1] = {\oplus}\{\,\m{cons} : \tock (\Box A \tensor \tock \Next^{k+4}\, \m{list}_A[n])\, \}$ \\[1ex]
  $\m{folder}_{AB} = {\with}\{\,\m{next} : \tock (\Box A \lolli \tock (B \lolli \tock \Next^k (\Next^5 B \tensor
\tock \Next^2\, \m{folder}_{AB}))), \m{done} : \tock \one\,\}$ \\[1ex]
  $l : \m{list}_A[n], f : \Next^2\, \m{folder}_{AB}, b : \Next^4 B \vdash \mi{fold} :: (r : \Next^{(k+5)n+4}\, B)$
\end{sill}
The type of $\mi{fold}$ indicates that if the combine function of
$\m{folder}_{AB}$ takes $k$ time units to compute, the result $r$ is
produced after $(k+5)n+4$ time units in the $\mc{RS}$ cost model.

\section{Further Related Work}
\label{sec:related}



In addition to the related work already mentioned, we highlight a few
related threads of research.

\paragraph{Session types and process calculi.}
In addition to the work on timed multiparty session
types~\cite{TMSTConcur14,Neykova14beat}, time has been introduced into
the $\pi$-calculus (see, for example, \citet{SaeedloeiTPC13}) or
session-based communication primitives (see, for example,
\citet{Lopez09places}) but generally these works do not develop a type
system.  \citet{Kobayashi02iandc} extends a (synchronous)
$\pi$-calculus with means to count parallel reduction steps.  He then
provides a type system to verify time-boundedness.  This is more
general in some dimension than our work because of a more permissive
underlying type and usage system, but it lacks internal and external
choice, genericity in the cost model, and provides bounds rather than
a fine gradation between exact and indefinite times.
Session types can also be derived by a Curry-Howard interpretation of
\emph{classical linear logic}~\citep{Wadler12icfp} but we are not
aware of temporal extensions. We conjecture that there is a classical
version of our system where $\Box$ and $\Dia$ are dual and $\Next$ is
self-dual.


\paragraph{Reactive programming.}
Synchronous dataflow languages such as Lustre~\citep{LustreIEEE},
Esterel~\citep{Esterel92}, or Lucid
Synchrone~\citep{Pouzet06lucidsynchrone} are time-synchronous with
uni-directional flow and thus may be compared to the fragment of our
language with internal choice ($\oplus$) and the next-time modality
($\Next A$), augmented with existential quantification over basic data
values like booleans and integers (which we have omitted here only for
the sake of brevity).  The global clock would map to our underlying
notion of time, but data-dependent local clocks would have to be
encoded at a relatively low level using streams of option type,
compromising the brevity and elegance of these languages.
Furthermore, synchronous dataflow languages generally permit sharing
of channels, which, although part of many session-typed
languages~\cite{Caires10concur,Balzer17icfp}, require further
investigation in our setting. On the other hand, we support a number
of additional types such as external choice ($\with$) for
bidirectional communication and higher-order channel-passing
($A \lolli B$, $A \tensor B$).
In the context of functional reactive programming, a
Nakano-style~\cite{Nakano00} temporal modality has been used to ensure
productivity~\cite{KrishnaswamiB11}. A difference in our work is that
we consider concurrent processes and that our types prescribe the
timing of messages.

\paragraph{Computational interpretations of $\Next A$.} A first
computational interpretation of the next-time modality under a
proofs-as-programs paradigm was given by \citet{Davies96lics}.  The
basis is natural deduction for a (non-linear!) intutionistic
linear-time temporal logic with only the next-time modality. Rather
than capturing cost, the programmer could indicate \emph{staging} by
stipulating that some subexpressions should be evaluated ``at the next
time''.  The natural operational semantics then is a
logically-motivated form of \emph{partial evaluation} which yields a
residual program of type $\Next A$. This idea was picked up
by \citet{Feltman16esop} to instead \emph{split} the program
statically into two stages where results from the first stage are
communicated to the second. Again, neither linearity (in the sense of
linear logic), nor any specific cost semantics appears in this work.

\paragraph{Other techniques.}
Inferring the cost of concurrent programs is a fundamental
problem in resource analysis. \citet{HoffmannESOP15}
introduce the first automatic analysis for deriving bounds
on the worst-case evaluation cost of parallel first-order
functional programs. Their main limitation is that they
can only handle parallel computation; they don't support
message-passing or shared memory based concurrency.
\citet{BlellochPipelining97} use pipelining~\citep{PaulPipelining}
to improve the complexity of parallel algorithms. However,
they use futures~\citep{Halstead:MULTILISP85}, a parallel
language construct to implement pipelining without the
programmer having to specify them explicitly. The runtime of
algorithms is determined by analyzing the work and depth
in a language-based cost model. The work relates to ours
in the sense that pipelines can have delays, which can be data
dependent. However, the algorithms they analyze have no
message-passing concurrency or other synchronization
constructs. \citet{AlbertSAS15} devised a static analysis for
inferring the parallel cost of distributed systems. They first
perform a block-level analysis to estimate the serial cost,
then construct a distributed flow graph (DFG) to capture the
parallelism and then obtain the parallel cost by computing the
maximal cost path in the DFG. However, the bounds they
produce are modulo a points-to and serial cost analysis. Hence,
an imprecise points-to analysis will result in imprecise parallel
cost bounds. Moreover, since their technique is based on
static analysis, it is not compositional and a whole program
analysis is needed to infer bounds on each module.
Recently, a bounded linear typing discipline~\citep{GhicaESOP14}
modeled in a semiring was proposed for resource-sensitive
compilation. It was then used to calculate and control execution
time in a higher-order functional programming language.
However, this language did not support recursion.


\section{Conclusion}

We have developed a system of temporal session types that can
accommodate and analyze concurrent programs with respect to a variety
of different cost models. Types can vary in precision, based on
desired and available information, and includes latency, rate,
response time, and span of computations. It is constructed in a
modular way, on top of a system of basic session types, and therefore
lends itself to easy generalization. We have illustrated the type
system through a number of simple programs on streams of bits, binary
counters, lists, stacks, queues, and trees.  Time reconstruction and
subtyping go some way towards alleviating demands on the programmer
and supporting program reuse. In ongoing work we are exploring an
implementation with an eye toward practical aspects of time
reconstruction and, beyond that, automatic resource analysis based on
internal measures of processes such as the length of a list or the
height of a tree---so far, we have carried out these analyses by hand.

\bibliography{fp,refs}


\begin{thebibliography}{43}


\ifx \showCODEN    \undefined \def \showCODEN     #1{\unskip}     \fi
\ifx \showDOI      \undefined \def \showDOI       #1{#1}\fi
\ifx \showISBNx    \undefined \def \showISBNx     #1{\unskip}     \fi
\ifx \showISBNxiii \undefined \def \showISBNxiii  #1{\unskip}     \fi
\ifx \showISSN     \undefined \def \showISSN      #1{\unskip}     \fi
\ifx \showLCCN     \undefined \def \showLCCN      #1{\unskip}     \fi
\ifx \shownote     \undefined \def \shownote      #1{#1}          \fi
\ifx \showarticletitle \undefined \def \showarticletitle #1{#1}   \fi
\ifx \showURL      \undefined \def \showURL       {\relax}        \fi
\providecommand\bibfield[2]{#2}
\providecommand\bibinfo[2]{#2}
\providecommand\natexlab[1]{#1}
\providecommand\showeprint[2][]{arXiv:#2}

\bibitem[\protect\citeauthoryear{Albert, Correas, Johnsen, and
  Rom{\'a}n-D{\'i}ez}{Albert et~al\mbox{.}}{2015}]%
        {AlbertSAS15}
\bibfield{author}{\bibinfo{person}{Elvira Albert}, \bibinfo{person}{Jes{\'u}s
  Correas}, \bibinfo{person}{Einar~Broch Johnsen}, {and}
  \bibinfo{person}{Guillermo Rom{\'a}n-D{\'i}ez}.}
  \bibinfo{year}{2015}\natexlab{}.
\newblock \showarticletitle{Parallel Cost Analysis of Distributed Systems}. In
  \bibinfo{booktitle}{\emph{Static Analysis}},
  \bibfield{editor}{\bibinfo{person}{Sandrine Blazy} {and}
  \bibinfo{person}{Thomas Jensen}} (Eds.). \bibinfo{publisher}{Springer Berlin
  Heidelberg}, \bibinfo{address}{Berlin, Heidelberg},
  \bibinfo{pages}{275--292}.
\newblock
\showISBNx{978-3-662-48288-9}


\bibitem[\protect\citeauthoryear{Avanzini, Lago, and Moser}{Avanzini
  et~al\mbox{.}}{2015}]%
        {AvanziniLM15}
\bibfield{author}{\bibinfo{person}{Martin Avanzini}, \bibinfo{person}{Ugo~Dal
  Lago}, {and} \bibinfo{person}{Georg Moser}.} \bibinfo{year}{2015}\natexlab{}.
\newblock \showarticletitle{{Analysing the Complexity of Functional Programs:
  Higher-Order Meets First-Order}}. In \bibinfo{booktitle}{\emph{{29th Int.
  Conf. on Functional Programming (ICFP'15)}}}.
\newblock


\bibitem[\protect\citeauthoryear{Balzer and Pfenning}{Balzer and
  Pfenning}{2017}]%
        {Balzer17icfp}
\bibfield{author}{\bibinfo{person}{Stephanie Balzer} {and}
  \bibinfo{person}{Frank Pfenning}.} \bibinfo{year}{2017}\natexlab{}.
\newblock \showarticletitle{Manifest Sharing with Session Types}. In
  \bibinfo{booktitle}{\emph{International Conference on Functional Programming
  (ICFP)}}. \bibinfo{publisher}{ACM}, \bibinfo{pages}{37:1--37:29}.
\newblock


\bibitem[\protect\citeauthoryear{Berry and Gonthier}{Berry and
  Gonthier}{1992}]%
        {Esterel92}
\bibfield{author}{\bibinfo{person}{G{\'e}rard Berry} {and}
  \bibinfo{person}{Georges Gonthier}.} \bibinfo{year}{1992}\natexlab{}.
\newblock \showarticletitle{The ESTEREL Synchronous Programming Language:
  Design, Semantics, Implementation}.
\newblock \bibinfo{journal}{\emph{Sci. Comput. Program.}} \bibinfo{volume}{19},
  \bibinfo{number}{2} (\bibinfo{date}{Nov.} \bibinfo{year}{1992}),
  \bibinfo{pages}{87--152}.
\newblock
\showISSN{0167-6423}
\urldef\tempurl%
\url{https://doi.org/10.1016/0167-6423(92)90005-V}
\showDOI{\tempurl}


\bibitem[\protect\citeauthoryear{Blelloch and Reid-Miller}{Blelloch and
  Reid-Miller}{1997}]%
        {BlellochPipelining97}
\bibfield{author}{\bibinfo{person}{Guy~E. Blelloch} {and}
  \bibinfo{person}{Margaret Reid-Miller}.} \bibinfo{year}{1997}\natexlab{}.
\newblock \showarticletitle{Pipelining with Futures}. In
  \bibinfo{booktitle}{\emph{Proceedings of the Ninth Annual ACM Symposium on
  Parallel Algorithms and Architectures}} \emph{(\bibinfo{series}{SPAA '97})}.
  \bibinfo{publisher}{ACM}, \bibinfo{address}{New York, NY, USA},
  \bibinfo{pages}{249--259}.
\newblock
\showISBNx{0-89791-890-8}
\urldef\tempurl%
\url{https://doi.org/10.1145/258492.258517}
\showDOI{\tempurl}


\bibitem[\protect\citeauthoryear{Bocchi, Yang, and Yoshida}{Bocchi
  et~al\mbox{.}}{2014}]%
        {TMSTConcur14}
\bibfield{author}{\bibinfo{person}{Laura Bocchi}, \bibinfo{person}{Weizhen
  Yang}, {and} \bibinfo{person}{Nobuko Yoshida}.}
  \bibinfo{year}{2014}\natexlab{}.
\newblock \showarticletitle{Timed Multiparty Session Types}. In
  \bibinfo{booktitle}{\emph{CONCUR 2014 -- Concurrency Theory}},
  \bibfield{editor}{\bibinfo{person}{Paolo Baldan} {and}
  \bibinfo{person}{Daniele Gorla}} (Eds.). \bibinfo{publisher}{Springer Berlin
  Heidelberg}, \bibinfo{address}{Berlin, Heidelberg},
  \bibinfo{pages}{419--434}.
\newblock
\showISBNx{978-3-662-44584-6}


\bibitem[\protect\citeauthoryear{Caires, P{\'e}rez, Pfenning, and
  Toninho}{Caires et~al\mbox{.}}{2013}]%
        {Caires13esop}
\bibfield{author}{\bibinfo{person}{Lu{\'\i}s Caires}, \bibinfo{person}{Jorge~A.
  P{\'e}rez}, \bibinfo{person}{Frank Pfenning}, {and} \bibinfo{person}{Bernardo
  Toninho}.} \bibinfo{year}{2013}\natexlab{}.
\newblock \showarticletitle{Behavioral Polymorphism and Parametricity in
  Session-Based Communication}. In \bibinfo{booktitle}{\emph{Proceedings of the
  European Symposium on Programming (ESOP'13)}},
  \bibfield{editor}{\bibinfo{person}{M.Felleisen} {and}
  \bibinfo{person}{P.Gardner}} (Eds.). \bibinfo{publisher}{Springer LNCS 7792},
  \bibinfo{address}{Rome, Italy}, \bibinfo{pages}{330--349}.
\newblock


\bibitem[\protect\citeauthoryear{Caires and Pfenning}{Caires and
  Pfenning}{2010}]%
        {Caires10concur}
\bibfield{author}{\bibinfo{person}{Lu{\'\i}s Caires} {and}
  \bibinfo{person}{Frank Pfenning}.} \bibinfo{year}{2010}\natexlab{}.
\newblock \showarticletitle{Session Types as Intuitionistic Linear
  Propositions}. In \bibinfo{booktitle}{\emph{Proceedings of the 21st
  International Conference on Concurrency Theory (CONCUR 2010)}}.
  \bibinfo{publisher}{Springer LNCS 6269}, \bibinfo{address}{Paris, France},
  \bibinfo{pages}{222--236}.
\newblock


\bibitem[\protect\citeauthoryear{Caires, Pfenning, and Toninho}{Caires
  et~al\mbox{.}}{2016}]%
        {Caires16mscs}
\bibfield{author}{\bibinfo{person}{Lu{\'\i}s Caires}, \bibinfo{person}{Frank
  Pfenning}, {and} \bibinfo{person}{Bernardo Toninho}.}
  \bibinfo{year}{2016}\natexlab{}.
\newblock \showarticletitle{Linear Logic Propositions as Session Types}.
\newblock \bibinfo{journal}{\emph{Mathematical Structures in Computer Science}}
  \bibinfo{volume}{26}, \bibinfo{number}{3} (\bibinfo{year}{2016}),
  \bibinfo{pages}{367--423}.
\newblock


\bibitem[\protect\citeauthoryear{Cervesato and Scedrov}{Cervesato and
  Scedrov}{2009}]%
        {Cervesato09ic}
\bibfield{author}{\bibinfo{person}{Iliano Cervesato} {and}
  \bibinfo{person}{Andre Scedrov}.} \bibinfo{year}{2009}\natexlab{}.
\newblock \showarticletitle{Relating State-Based and Process-Based Concurrency
  through Linear Logic}.
\newblock \bibinfo{journal}{\emph{Information and Computation}}
  \bibinfo{volume}{207}, \bibinfo{number}{10} (\bibinfo{date}{Oct.}
  \bibinfo{year}{2009}), \bibinfo{pages}{1044--1077}.
\newblock


\bibitem[\protect\citeauthoryear{Danner, Licata, and Ramyaa}{Danner
  et~al\mbox{.}}{2015}]%
        {DannerLR15}
\bibfield{author}{\bibinfo{person}{Norman Danner}, \bibinfo{person}{Daniel~R.
  Licata}, {and} \bibinfo{person}{Ramyaa Ramyaa}.}
  \bibinfo{year}{2015}\natexlab{}.
\newblock \showarticletitle{{Denotational Cost Semantics for Functional
  Languages with Inductive Types}}. In \bibinfo{booktitle}{\emph{{29th Int.
  Conf. on Functional Programming (ICFP'15)}}}.
\newblock


\bibitem[\protect\citeauthoryear{Das, Hoffmann, and Pfenning}{Das
  et~al\mbox{.}}{2017}]%
        {Das18arxiv}
\bibfield{author}{\bibinfo{person}{Ankush Das}, \bibinfo{person}{Jan Hoffmann},
  {and} \bibinfo{person}{Frank Pfenning}.} \bibinfo{year}{2017}\natexlab{}.
\newblock \showarticletitle{Work Analysis with Resource-Aware Session Types}.
\newblock \bibinfo{journal}{\emph{CoRR}}  \bibinfo{volume}{abs/1712.08310}
  (\bibinfo{year}{2017}).
\newblock
\showeprint[arxiv]{1712.08310}
\urldef\tempurl%
\url{http://arxiv.org/abs/1712.08310}
\showURL{%
\tempurl}


\bibitem[\protect\citeauthoryear{Davies}{Davies}{1996}]%
        {Davies96lics}
\bibfield{author}{\bibinfo{person}{Rowan Davies}.}
  \bibinfo{year}{1996}\natexlab{}.
\newblock \showarticletitle{A Temporal Logic Approach to Binding-Time
  Analysis}. In \bibinfo{booktitle}{\emph{Proceedings of the Eleventh Annual
  Symposium on Logic in Computer Science}},
  \bibfield{editor}{\bibinfo{person}{E.~Clarke}} (Ed.).
  \bibinfo{publisher}{IEEE Computer Society Press}, \bibinfo{address}{New
  Brunswick, New Jersey}, \bibinfo{pages}{184--195}.
\newblock
\urldef\tempurl%
\url{http://www.cs.cmu.edu/afs/cs/user/rowan/www/papers/multbta.ps.Z}
\showURL{%
\tempurl}


\bibitem[\protect\citeauthoryear{Feltman, Angiuli, Acar, and
  Fatahalian}{Feltman et~al\mbox{.}}{2016}]%
        {Feltman16esop}
\bibfield{author}{\bibinfo{person}{Nicolas Feltman}, \bibinfo{person}{Carlo
  Angiuli}, \bibinfo{person}{Umut Acar}, {and} \bibinfo{person}{Kayvon
  Fatahalian}.} \bibinfo{year}{2016}\natexlab{}.
\newblock \showarticletitle{Automatically Splitting a Two-Stage Lambda
  Calculus}. In \bibinfo{booktitle}{\emph{Proceedings of the 25th European
  Symposium on Programming (ESOP)}},
  \bibfield{editor}{\bibinfo{person}{P.~Thiemann}} (Ed.).
  \bibinfo{publisher}{Springer LNCS 9632}, \bibinfo{address}{Eindhoven, The
  Netherlands}, \bibinfo{pages}{255--281}.
\newblock


\bibitem[\protect\citeauthoryear{Fortier and Santocanale}{Fortier and
  Santocanale}{2013}]%
        {Fortier13csl}
\bibfield{author}{\bibinfo{person}{J{\'e}r{\^o}me Fortier} {and}
  \bibinfo{person}{Luigi Santocanale}.} \bibinfo{year}{2013}\natexlab{}.
\newblock \showarticletitle{Cuts for Circular Proofs: {S}emantics and Cut
  Elimination}. In \bibinfo{booktitle}{\emph{22nd Conference on Computer
  Science Logic}} \emph{(\bibinfo{series}{LIPIcs})}, Vol.~\bibinfo{volume}{23}.
  \bibinfo{pages}{248--262}.
\newblock


\bibitem[\protect\citeauthoryear{Gay and Hole}{Gay and Hole}{2005}]%
        {Gay05acta}
\bibfield{author}{\bibinfo{person}{Simon~J. Gay} {and} \bibinfo{person}{Malcolm
  Hole}.} \bibinfo{year}{2005}\natexlab{}.
\newblock \showarticletitle{Subtyping for Session Types in the
  {$\pi$}-Calculus}.
\newblock \bibinfo{journal}{\emph{Acta Informatica}} \bibinfo{volume}{42},
  \bibinfo{number}{2--3} (\bibinfo{year}{2005}), \bibinfo{pages}{191--225}.
\newblock


\bibitem[\protect\citeauthoryear{Ghica and Smith}{Ghica and Smith}{2014}]%
        {GhicaESOP14}
\bibfield{author}{\bibinfo{person}{Dan~R. Ghica} {and} \bibinfo{person}{Alex~I.
  Smith}.} \bibinfo{year}{2014}\natexlab{}.
\newblock \showarticletitle{Bounded Linear Types in a Resource Semiring}. In
  \bibinfo{booktitle}{\emph{Proceedings of the 23rd European Symposium on
  Programming Languages and Systems - Volume 8410}}.
  \bibinfo{publisher}{Springer-Verlag New York, Inc.}, \bibinfo{address}{New
  York, NY, USA}, \bibinfo{pages}{331--350}.
\newblock
\showISBNx{978-3-642-54832-1}
\urldef\tempurl%
\url{https://doi.org/10.1007/978-3-642-54833-8_18}
\showDOI{\tempurl}


\bibitem[\protect\citeauthoryear{Gimenez and Moser}{Gimenez and Moser}{2016}]%
        {GimenezPOPL16}
\bibfield{author}{\bibinfo{person}{St{\'e}phane Gimenez} {and}
  \bibinfo{person}{Georg Moser}.} \bibinfo{year}{2016}\natexlab{}.
\newblock \showarticletitle{The Complexity of Interaction}. In
  \bibinfo{booktitle}{\emph{Proceedings of the 43rd Annual ACM SIGPLAN-SIGACT
  Symposium on Principles of Programming Languages}}
  \emph{(\bibinfo{series}{POPL '16})}. \bibinfo{publisher}{ACM},
  \bibinfo{address}{New York, NY, USA}, \bibinfo{pages}{243--255}.
\newblock
\showISBNx{978-1-4503-3549-2}
\urldef\tempurl%
\url{https://doi.org/10.1145/2837614.2837646}
\showDOI{\tempurl}


\bibitem[\protect\citeauthoryear{Griffith}{Griffith}{2016}]%
        {Griffith16phd}
\bibfield{author}{\bibinfo{person}{Dennis Griffith}.}
  \bibinfo{year}{2016}\natexlab{}.
\newblock \emph{\bibinfo{title}{Polarized Substructural Session Types}}.
\newblock \bibinfo{thesistype}{Ph.D. Dissertation}. \bibinfo{school}{University
  of Illinois at Urbana-Champaign}.
\newblock


\bibitem[\protect\citeauthoryear{Griffith and Gunter}{Griffith and
  Gunter}{2013}]%
        {Griffith13nasa}
\bibfield{author}{\bibinfo{person}{Dennis Griffith} {and}
  \bibinfo{person}{Elsa~L. Gunter}.} \bibinfo{year}{2013}\natexlab{}.
\newblock \showarticletitle{Liquid Pi: Inferrable Dependent Session Types}. In
  \bibinfo{booktitle}{\emph{Proceedings of the NASA Formal Methods Symposium}}.
  \bibinfo{publisher}{Springer LNCS 7871}, \bibinfo{pages}{186--197}.
\newblock


\bibitem[\protect\citeauthoryear{Gulwani, Mehra, and Chilimbi}{Gulwani
  et~al\mbox{.}}{2009}]%
        {GulwaniMC09}
\bibfield{author}{\bibinfo{person}{Sumit Gulwani}, \bibinfo{person}{Krishna~K.
  Mehra}, {and} \bibinfo{person}{Trishul~M. Chilimbi}.}
  \bibinfo{year}{2009}\natexlab{}.
\newblock \showarticletitle{{SPEED: Precise and Efficient Static Estimation of
  Program Computational Complexity}}. In \bibinfo{booktitle}{\emph{{36th ACM
  Symp. on Principles of Prog. Langs. (POPL'09)}}}. \bibinfo{pages}{127--139}.
\newblock


\bibitem[\protect\citeauthoryear{Halbwachs, Caspi, Raymond, and
  Pilaud}{Halbwachs et~al\mbox{.}}{1991}]%
        {LustreIEEE}
\bibfield{author}{\bibinfo{person}{N. Halbwachs}, \bibinfo{person}{P. Caspi},
  \bibinfo{person}{P. Raymond}, {and} \bibinfo{person}{D. Pilaud}.}
  \bibinfo{year}{1991}\natexlab{}.
\newblock \showarticletitle{The synchronous data flow programming language
  LUSTRE}.
\newblock \bibinfo{journal}{\emph{Proc. IEEE}} \bibinfo{volume}{79},
  \bibinfo{number}{9} (\bibinfo{date}{Sep} \bibinfo{year}{1991}),
  \bibinfo{pages}{1305--1320}.
\newblock
\showISSN{0018-9219}
\urldef\tempurl%
\url{https://doi.org/10.1109/5.97300}
\showDOI{\tempurl}


\bibitem[\protect\citeauthoryear{Halstead}{Halstead}{1985}]%
        {Halstead:MULTILISP85}
\bibfield{author}{\bibinfo{person}{Robert~H. Halstead, Jr.}}
  \bibinfo{year}{1985}\natexlab{}.
\newblock \showarticletitle{MULTILISP: A Language for Concurrent Symbolic
  Computation}.
\newblock \bibinfo{journal}{\emph{ACM Trans. Program. Lang. Syst.}}
  \bibinfo{volume}{7}, \bibinfo{number}{4} (\bibinfo{date}{Oct.}
  \bibinfo{year}{1985}), \bibinfo{pages}{501--538}.
\newblock
\showISSN{0164-0925}
\urldef\tempurl%
\url{https://doi.org/10.1145/4472.4478}
\showDOI{\tempurl}


\bibitem[\protect\citeauthoryear{Hoffmann, Das, and Weng}{Hoffmann
  et~al\mbox{.}}{2017}]%
        {HoffmannW15}
\bibfield{author}{\bibinfo{person}{Jan Hoffmann}, \bibinfo{person}{Ankush Das},
  {and} \bibinfo{person}{Shu-Chun Weng}.} \bibinfo{year}{2017}\natexlab{}.
\newblock \showarticletitle{{Towards Automatic Resource Bound Analysis for
  OCaml}}. In \bibinfo{booktitle}{\emph{44th Symposium on Principles of
  Programming Languages (POPL'17)}}.
\newblock


\bibitem[\protect\citeauthoryear{Hoffmann and Shao}{Hoffmann and Shao}{2015}]%
        {HoffmannESOP15}
\bibfield{author}{\bibinfo{person}{Jan Hoffmann} {and} \bibinfo{person}{Zhong
  Shao}.} \bibinfo{year}{2015}\natexlab{}.
\newblock \showarticletitle{Automatic Static Cost Analysis for Parallel
  Programs}. In \bibinfo{booktitle}{\emph{Proceedings of the 24th European
  Symposium on Programming on Programming Languages and Systems - Volume
  9032}}. \bibinfo{publisher}{Springer-Verlag New York, Inc.},
  \bibinfo{address}{New York, NY, USA}, \bibinfo{pages}{132--157}.
\newblock
\showISBNx{978-3-662-46668-1}
\urldef\tempurl%
\url{https://doi.org/10.1007/978-3-662-46669-8_6}
\showDOI{\tempurl}


\bibitem[\protect\citeauthoryear{Honda, Vasconcelos, and Kubo}{Honda
  et~al\mbox{.}}{1998}]%
        {Honda98esop}
\bibfield{author}{\bibinfo{person}{Kohei Honda}, \bibinfo{person}{Vasco~T.
  Vasconcelos}, {and} \bibinfo{person}{Makoto Kubo}.}
  \bibinfo{year}{1998}\natexlab{}.
\newblock \showarticletitle{Language Primitives and Type Discipline for
  Structured Communication-Based Programming}. In \bibinfo{booktitle}{\emph{7th
  European Symposium on Programming Languages and Systems}}
  \emph{(\bibinfo{series}{ESOP'98})}. \bibinfo{publisher}{Springer LNCS 1381},
  \bibinfo{pages}{122--138}.
\newblock


\bibitem[\protect\citeauthoryear{Kobayashi}{Kobayashi}{2002}]%
        {Kobayashi02iandc}
\bibfield{author}{\bibinfo{person}{Naoki Kobayashi}.}
  \bibinfo{year}{2002}\natexlab{}.
\newblock \showarticletitle{A Type System for Lock-Free Processes}.
\newblock \bibinfo{journal}{\emph{Information and Computation}}
  \bibinfo{volume}{177} (\bibinfo{year}{2002}), \bibinfo{pages}{122--159}.
\newblock


\bibitem[\protect\citeauthoryear{Krishnaswami and Benton}{Krishnaswami and
  Benton}{2011}]%
        {KrishnaswamiB11}
\bibfield{author}{\bibinfo{person}{Neelakantan~R. Krishnaswami} {and}
  \bibinfo{person}{Nick Benton}.} \bibinfo{year}{2011}\natexlab{}.
\newblock \showarticletitle{{Ultrametric Semantics of Reactive Programs}}. In
  \bibinfo{booktitle}{\emph{26th {IEEE} Symposium on Logic in Computer Science,
  {(LICS'11)}}}. \bibinfo{pages}{257--266}.
\newblock


\bibitem[\protect\citeauthoryear{Lago and Gaboardi}{Lago and Gaboardi}{2011}]%
        {LagoG11}
\bibfield{author}{\bibinfo{person}{Ugo~Dal Lago} {and} \bibinfo{person}{Marco
  Gaboardi}.} \bibinfo{year}{2011}\natexlab{}.
\newblock \showarticletitle{{Linear Dependent Types and Relative
  Completeness}}. In \bibinfo{booktitle}{\emph{{26th IEEE Symp. on Logic in
  Computer Science (LICS'11)}}}. \bibinfo{pages}{133--142}.
\newblock


\bibitem[\protect\citeauthoryear{Lange and Yoshida}{Lange and Yoshida}{2017}]%
        {Lange17fossacs}
\bibfield{author}{\bibinfo{person}{Julien Lange} {and} \bibinfo{person}{Nobuko
  Yoshida}.} \bibinfo{year}{2017}\natexlab{}.
\newblock \showarticletitle{On the Undecidability of Asynchronous Session
  Subtyping}. In \bibinfo{booktitle}{\emph{Proceedings of the 20th
  International Conference on Foundations of Software Science and Computation
  Structures (FoSSaCS)}}. \bibinfo{publisher}{Springer LNCS 10203},
  \bibinfo{pages}{441--457}.
\newblock


\bibitem[\protect\citeauthoryear{L{\'o}pez, Olarte, and P{\'e}rez}{L{\'o}pez
  et~al\mbox{.}}{2009}]%
        {Lopez09places}
\bibfield{author}{\bibinfo{person}{Hugo~A. L{\'o}pez}, \bibinfo{person}{Carlos
  Olarte}, {and} \bibinfo{person}{Jorge~A. P{\'e}rez}.}
  \bibinfo{year}{2009}\natexlab{}.
\newblock \showarticletitle{Towards a Unified Framework for Declarative
  Structure Communications}. In \bibinfo{booktitle}{\emph{Proceedings of the
  Workshop on Programming Language Approaches to Concurrency and
  Communication-Centric Software (PLACES)}},
  \bibfield{editor}{\bibinfo{person}{A.~Beresford} {and}
  \bibinfo{person}{S.~Gay}} (Eds.). \bibinfo{publisher}{EPTCS 17},
  \bibinfo{pages}{1--15}.
\newblock


\bibitem[\protect\citeauthoryear{Nakano}{Nakano}{2000}]%
        {Nakano00}
\bibfield{author}{\bibinfo{person}{Hiroshi Nakano}.}
  \bibinfo{year}{2000}\natexlab{}.
\newblock \showarticletitle{{A Modality for Recursion}}. In
  \bibinfo{booktitle}{\emph{15th {IEEE} Symposium on Logic in Computer Science
  (LICS'00)}}. \bibinfo{pages}{255--266}.
\newblock


\bibitem[\protect\citeauthoryear{Neykova, Bocchi, and Yoshida}{Neykova
  et~al\mbox{.}}{2014}]%
        {Neykova14beat}
\bibfield{author}{\bibinfo{person}{Rumyana Neykova}, \bibinfo{person}{Laura
  Bocchi}, {and} \bibinfo{person}{Nobuko Yoshida}.}
  \bibinfo{year}{2014}\natexlab{}.
\newblock \showarticletitle{Timed Runtime Monitoring for Multiparty
  Conversations}. In \bibinfo{booktitle}{\emph{3rd International Workshop on
  Behavioural Types (BEAT 2014)}}.
\newblock


\bibitem[\protect\citeauthoryear{Paul, Vishkin, and Wagener}{Paul
  et~al\mbox{.}}{1983}]%
        {PaulPipelining}
\bibfield{author}{\bibinfo{person}{W. Paul}, \bibinfo{person}{U. Vishkin},
  {and} \bibinfo{person}{H. Wagener}.} \bibinfo{year}{1983}\natexlab{}.
\newblock \showarticletitle{Parallel dictionaries on 2--3 trees}. In
  \bibinfo{booktitle}{\emph{Automata, Languages and Programming}},
  \bibfield{editor}{\bibinfo{person}{Josep Diaz}} (Ed.).
  \bibinfo{publisher}{Springer Berlin Heidelberg}, \bibinfo{address}{Berlin,
  Heidelberg}, \bibinfo{pages}{597--609}.
\newblock
\showISBNx{978-3-540-40038-7}


\bibitem[\protect\citeauthoryear{Pfenning and Griffith}{Pfenning and
  Griffith}{2015}]%
        {Pfenning15fossacs}
\bibfield{author}{\bibinfo{person}{Frank Pfenning} {and}
  \bibinfo{person}{Dennis Griffith}.} \bibinfo{year}{2015}\natexlab{}.
\newblock \showarticletitle{Polarized Substructural Session Types}. In
  \bibinfo{booktitle}{\emph{Proceedings of the 18th International Conference on
  Foundations of Software Science and Computation Structures (FoSSaCS 2015)}},
  \bibfield{editor}{\bibinfo{person}{A.~Pitts}} (Ed.).
  \bibinfo{publisher}{Springer LNCS 9034}, \bibinfo{address}{London, England},
  \bibinfo{pages}{3--22}.
\newblock
\newblock
\shownote{Invited talk.}


\bibitem[\protect\citeauthoryear{Pnueli}{Pnueli}{1977}]%
        {Pnueli77}
\bibfield{author}{\bibinfo{person}{Amir Pnueli}.}
  \bibinfo{year}{1977}\natexlab{}.
\newblock \showarticletitle{The Temporal Logic of Programs}. In
  \bibinfo{booktitle}{\emph{Proceedings of the 18th Symposium on Foundations of
  Computer Science (FOCS'77)}}. \bibinfo{publisher}{IEEE Computer Society},
  \bibinfo{pages}{46--57}.
\newblock


\bibitem[\protect\citeauthoryear{Pouzet}{Pouzet}{2006}]%
        {Pouzet06lucidsynchrone}
\bibfield{author}{\bibinfo{person}{Marc Pouzet}.}
  \bibinfo{year}{2006}\natexlab{}.
\newblock \bibinfo{title}{Lucid Synchrone Release, version 3.0 Tutorial and
  Reference Manual}.
\newblock   (\bibinfo{year}{2006}).
\newblock


\bibitem[\protect\citeauthoryear{Saeedloei and Gupta}{Saeedloei and
  Gupta}{2014}]%
        {SaeedloeiTPC13}
\bibfield{author}{\bibinfo{person}{Neda Saeedloei} {and} \bibinfo{person}{Gopal
  Gupta}.} \bibinfo{year}{2014}\natexlab{}.
\newblock \showarticletitle{Timed $\pi$-Calculus}. In
  \bibinfo{booktitle}{\emph{8th International Symposium on Trustworthy Global
  Computing - Volume 8358}} \emph{(\bibinfo{series}{TGC 2013})}.
  \bibinfo{publisher}{Springer-Verlag New York, Inc.}, \bibinfo{address}{New
  York, NY, USA}, \bibinfo{pages}{119--135}.
\newblock
\showISBNx{978-3-319-05118-5}
\urldef\tempurl%
\url{https://doi.org/10.1007/978-3-319-05119-2_8}
\showDOI{\tempurl}


\bibitem[\protect\citeauthoryear{Silva, Florido, and Pfenning}{Silva
  et~al\mbox{.}}{2016}]%
        {Silva16linearity}
\bibfield{author}{\bibinfo{person}{Miguel Silva}, \bibinfo{person}{M{\'a}rio
  Florido}, {and} \bibinfo{person}{Frank Pfenning}.}
  \bibinfo{year}{2016}\natexlab{}.
\newblock \showarticletitle{Non-Blocking Concurrent Imperative Programming with
  Session Types}. In \bibinfo{booktitle}{\emph{Fourth International Workshop on
  Linearity}}.
\newblock


\bibitem[\protect\citeauthoryear{Toninho, Caires, and Pfenning}{Toninho
  et~al\mbox{.}}{2013}]%
        {Toninho13esop}
\bibfield{author}{\bibinfo{person}{Bernardo Toninho},
  \bibinfo{person}{Lu{\'\i}s Caires}, {and} \bibinfo{person}{Frank Pfenning}.}
  \bibinfo{year}{2013}\natexlab{}.
\newblock \showarticletitle{Higher-Order Processes, Functions, and Sessions: A
  Monadic Integration}. In \bibinfo{booktitle}{\emph{Proceedings of the
  European Symposium on Programming (ESOP'13)}},
  \bibfield{editor}{\bibinfo{person}{M.Felleisen} {and}
  \bibinfo{person}{P.Gardner}} (Eds.). \bibinfo{publisher}{Springer LNCS 7792},
  \bibinfo{address}{Rome, Italy}, \bibinfo{pages}{350--369}.
\newblock


\bibitem[\protect\citeauthoryear{Toninho, Caires, and Pfenning}{Toninho
  et~al\mbox{.}}{2014}]%
        {Toninho14tgc}
\bibfield{author}{\bibinfo{person}{Bernardo Toninho},
  \bibinfo{person}{Lu\'{\i}s Caires}, {and} \bibinfo{person}{Frank Pfenning}.}
  \bibinfo{year}{2014}\natexlab{}.
\newblock \showarticletitle{Corecursion and Non-Divergence in Session-Typed
  Processes}. In \bibinfo{booktitle}{\emph{Proceedings of the 9th International
  Symposium on Trustworthy Global Computing (TGC 2014)}},
  \bibfield{editor}{\bibinfo{person}{M.~Maffei} {and}
  \bibinfo{person}{E.~Tuosto}} (Eds.). \bibinfo{publisher}{Springer LNCS 8902},
  \bibinfo{address}{Rome, Italy}, \bibinfo{pages}{159--175}.
\newblock


\bibitem[\protect\citeauthoryear{Wadler}{Wadler}{2012}]%
        {Wadler12icfp}
\bibfield{author}{\bibinfo{person}{Philip Wadler}.}
  \bibinfo{year}{2012}\natexlab{}.
\newblock \showarticletitle{Propositions as Sessions}. In
  \bibinfo{booktitle}{\emph{Proceedings of the 17th International Conference on
  Functional Programming}} \emph{(\bibinfo{series}{ICFP 2012})}.
  \bibinfo{publisher}{ACM Press}, \bibinfo{address}{Copenhagen, Denmark},
  \bibinfo{pages}{273--286}.
\newblock


\bibitem[\protect\citeauthoryear{Çiçek, Barthe, Gaboardi, Garg, and
  Hoffmann}{Çiçek et~al\mbox{.}}{2017}]%
        {CicekBGGH16}
\bibfield{author}{\bibinfo{person}{Ezgi Çiçek}, \bibinfo{person}{Gilles
  Barthe}, \bibinfo{person}{Marco Gaboardi}, \bibinfo{person}{Deepak Garg},
  {and} \bibinfo{person}{Jan Hoffmann}.} \bibinfo{year}{2017}\natexlab{}.
\newblock \showarticletitle{{Relational Cost Analysis}}. In
  \bibinfo{booktitle}{\emph{44th Symposium on Principles of Programming
  Languages (POPL'17)}}.
\newblock


\end{thebibliography}

\end{document}